\documentclass[reprint,amsfonts, amssymb, amsmath,  showkeys, nofootinbib,pra, superscriptaddress, twocolumn,longbibliography]{revtex4-2}
\usepackage{float}
\makeatletter
\let\newfloat\newfloat@ltx
\makeatother
\usepackage[english]{babel}
\usepackage[utf8]{inputenc}
\usepackage{graphics}
\usepackage{selinput}

\usepackage{comment}
 
\usepackage{braket}
\usepackage{amsthm}
\usepackage{mathtools}
\usepackage{physics}
\usepackage{xcolor}
\usepackage{graphicx}
\usepackage[left=16mm,right=16mm,top=35mm,columnsep=15pt]{geometry} 
\usepackage{adjustbox}
\usepackage{placeins}
\usepackage[T1]{fontenc}
\usepackage{lipsum}
\usepackage{csquotes}

\usepackage{braket}
\usepackage{algorithm}
\usepackage[noend]{algpseudocode}

\def\LC{\mathcal{L}}

\def\ad{^{\dagger}}

\usepackage[makeroom]{cancel}
\usepackage[toc,page]{appendix}
\usepackage[colorlinks=true,citecolor=blue,linkcolor=magenta]{hyperref}

\newcommand{\dya}[1]{\ket{#1}\!\bra{#1}}

\newcommand{\poly}{\operatorname{poly}}

\newcommand{\BC}{\mathcal{B}}
\newcommand{\CC}{\mathcal{C}}
\newcommand{\DC}{\mathcal{D}}

\newcommand{\FC}{\mathcal{F}}
\newcommand{\GC}{\mathcal{G}}

\newcommand{\NC}{\mathcal{N}}
\newcommand{\OC}{\mathcal{O}}

\newcommand{\SC}{\mathcal{S}}

\newcommand{\UC}{\mathcal{U}}

\newcommand{\ZC}{\mathcal{Z}}

\renewcommand{\geq}{\geqslant}
\renewcommand{\leq}{\leqslant}

\renewcommand{\Re}{\text{Re}}

\renewcommand{\vec}[1]{\boldsymbol{#1}}  

\newcommand*{\id}{\openone}
\newcommand*{\iso}{\cong}

\newcommand{\bs}{\textsf{BS}}

\newcommand{\thv}{\vec{\theta}}
\newcommand{\delv}{\vec{\delta}}

\def\be{\begin{equation}}
\def\ee{\end{equation}}
\def\bs{\begin{split}}
\def\e{\end{split}}
\def\ba{\begin{eqnarray}}
\def\bea{\begin{eqnarray}}

\def\tea{\end{eqnarray}}
\def\ea{\end{eqnarray}}
\def\eea{\end{eqnarray}}

\def\liea{\mathfrak{k}}
\def\liea{\mathfrak{g}}

\def\liea{\mathfrak{k}}
\def\liea{\mathfrak{g}}

\newcommand\spn{\text{span}}

\newtheorem{theorem}{Theorem}
\newtheorem{lemma}{Lemma}

\newtheorem{definition}{Definition}

\usepackage{amssymb}
\usepackage{dsfont}

\def\be{\begin{equation}}
\def\te{\end{equation}}
\def\ee{\end{equation}}
\def\ba{\begin{eqnarray}}
\def\bea{\begin{eqnarray}}

\def\tea{\end{eqnarray}}
\def\ea{\end{eqnarray}}
\def\eea{\end{eqnarray}}

\begin{document}
\title{Effects of noise on the overparametrization of quantum neural networks}

\author{Diego Garc\'{i}a-Mart\'{i}n}
\affiliation{Information Sciences, Los Alamos National Laboratory, Los Alamos, NM 87545, USA}
\affiliation{Quantum Research Centre, Technology Innovation Institute, Abu Dhabi, UAE}
\affiliation{Instituto de Física Teórica, UAM/CSIC, Madrid, Spain}

\author{Mart\'{i}n Larocca}
\affiliation{Theoretical Division, Los Alamos National Laboratory, Los Alamos, New Mexico 87545, USA}
\affiliation{Center for Nonlinear Studies, Los Alamos National Laboratory, Los Alamos, New Mexico 87545, USA}

\author{M. Cerezo}
\thanks{cerezo@lanl.gov}
\affiliation{Information Sciences, Los Alamos National Laboratory, Los Alamos, NM 87545, USA}
\affiliation{Quantum Science Center, Oak Ridge, TN 37931, USA}

\begin{abstract}
Overparametrization is  one of the most surprising and notorious phenomena in machine learning. Recently, there  have been several efforts to study if, and how, Quantum Neural Networks (QNNs) acting in the absence of hardware noise can be overparametrized. In particular, it has been proposed that  a QNN can be defined as overparametrized if it has enough parameters to explore all available directions in state space. That is, if the rank of the Quantum Fisher Information Matrix (QFIM) for the QNN's output state is saturated. Here, we explore how the presence of noise affects the overparametrization phenomenon. Our results show that noise can ``turn on'' previously-zero eigenvalues of the QFIM. This enables the parametrized state to explore  directions that were otherwise inaccessible, thus potentially turning an overparametrized QNN into an underparametrized one. For small noise levels, the QNN is quasi-overparametrized, as large eigenvalues coexists with small ones. Then, we prove that as the magnitude of noise increases all the eigenvalues of the QFIM become exponentially suppressed, indicating that the state becomes insensitive to any change in the parameters.  As such,  there is a pull-and-tug effect where noise can enable new directions, but also suppress the sensitivity to parameter updates. Finally, our results imply that current QNN capacity measures are ill-defined when hardware noise is present.     
\end{abstract}

\maketitle

\section{Introduction}

Overparametrization has become one of the most important concepts for studying neural networks in classical machine learning. When a neural network is overparametrized, it has a capacity which is larger than the number of training points~\cite{zhang2021understanding}.  Despite being initially counterintuitive, as increasing the number of parameters can lead to overfitting, research has shown that overparametrization can actually improve the performance of a model~\cite{zhang2021understanding,allen2019convergence,du2019gradient,buhai2020empirical}. For example, it has been observed that the generalization error  can decrease when the model size is increased, a phenomenon known as double descent~\cite{belkin2019reconciling,advani2020high,geiger2019jamming}. Additionally, overparametrization can provide convergence guarantees, ensuring that a model will be able to find a good solution during its optimization~\cite{du2018gradient,brutzkus2018sgd}. These benefits make overparametrization an important consideration in the design of classical machine learning algorithms.

In the past few years, there has been a significant amount of effort towards merging concepts from classical machine learning with those of quantum computing, leading to the blossoming field of Quantum Machine Learning (QML)~\cite{biamonte2017quantum,cerezo2020variationalreview,bharti2021noisy,cerezo2022challenges}. The key idea here is that one can leverage the exponentially large dimension of the Hilbert space as a feature space to process and learn from data. Crucially, there is hope that QML has the potential of enabling a quantum advantage in the near-term~\cite{huang2021provably,huang2021quantum}. 

Within the framework of QML, parametrized quantum circuits, or Quantum Neural Networks (QNNs), have received considerable attention due to their versatility and wide usability~\cite{benedetti2019parameterized,nguyen2022atheory,cong2019quantum,beer2020training,verdon2017quantum}. While several works have studied the capabilities, trainability and performance of QNNs~\cite{larocca2021theory,wiersema2020exploring,kiani2020learning,you2022convergence,matos2022characterization,wang2022symmetric,abbas2020power,haug2021capacity,mcclean2018barren,cerezo2020cost,larocca2021diagnosing,zhang2020overparametrization,wierichs2020avoiding,funcke2021best,lee2021towards,anschuetz2021critical,bittel2021training,liu2022analytic,liu2021representation}, most of these consider noiseless scenarios which do not account for the effect of hardware noise~\cite{wang2020noise,franca2020limitations,sharma2019noise,xue2021effects,marshall2020characterizing}. However, since noise is an intrinsic element in near-term quantum computing, it is fundamental to understand how its presence alters noiseless results and changes our understanding of QNNs. For instance, it is known that there exist polynomial-time (yet unpractical) classical algorithms for simulating random quantum circuits in the presence of local depolarizing noise~\cite{gao2018efficient,aharonov2022polynomial}.

In this work we study how the recently developed understanding of overparametrization in QNNs~\cite{larocca2021theory,wiersema2020exploring,kiani2020learning,you2022convergence,matos2022characterization,wang2022symmetric} is affected by the presence of quantum noise. In particular, we will review the results of Ref.~\cite{larocca2021theory}, which characterizes the critical number of parameters needed to overparametrize a QNN. It has been observed that
underparametrized QNNs exhibit spurious local minima in the optimization landscape that hinder their trainability. By adding enough parameters to the circuit (hence overparametrizing it), these  false local traps disappear. Since the previous  facilitates the QNN's parameter training,  the overparametrization onset corresponds to a veritable computational phase transition. Notably, in Ref.~\cite{larocca2021theory} the number of parameters needed to  overparametrize a QNN is defined as those needed to saturate the rank of the Quantum Fisher Information Matrix (QFIM), and concomitantly the QNN's capacity, as introduced in~\cite{haug2021capacity,abbas2020power}.

Our results show that the presence of hardware noise can increase the rank of a QFIM whose rank would have been saturated in a noiseless scenario. That is, noise can turn null eigenvalues of a noiseless-state QFIM into non-null eigenvalues of the corresponding noisy-state  QFIM. As schematically depicted in Fig.~\ref{fig:main}, this means that hardware noise allows the QNN to explore previously unavailable directions. Hence,  some of the redundant parameters in an overparametrized noiseless QNN become relevant to control trajectories in state space when the effect of noise is accounted for. As such, noise can potentially render an overparametrized model into an underparametrized one. In addition, we analytically prove that as the noise strength (or depth of the circuit) increases, the eigenvalues of the QFIM become exponentially suppressed. Thus, for large noise levels (or for deep QNNs), the states become insensitive to any change in the parameters. On the positive side, our numerics show that for small noise levels, the model behaves as being \textit{quasi-overparametrized}: Large eigenvalues of the QFIM (the ones that are non-zero in the noiseless setting) coexists with small ones (the ones that were previously zero). Additionally, we prove that certain types of noise, specifically global depolarizing noise or measurement noise~\cite{sharma2019noise,maciejewski2020mitigation}, cannot increase the rank of the QFIM. To conclude, we discuss the implications of our results to QNN capacity measures proposed in the literature, and to other fields such as quantum metrology.

\begin{figure}[t!]
    \centering
    \includegraphics[width=.8\linewidth]{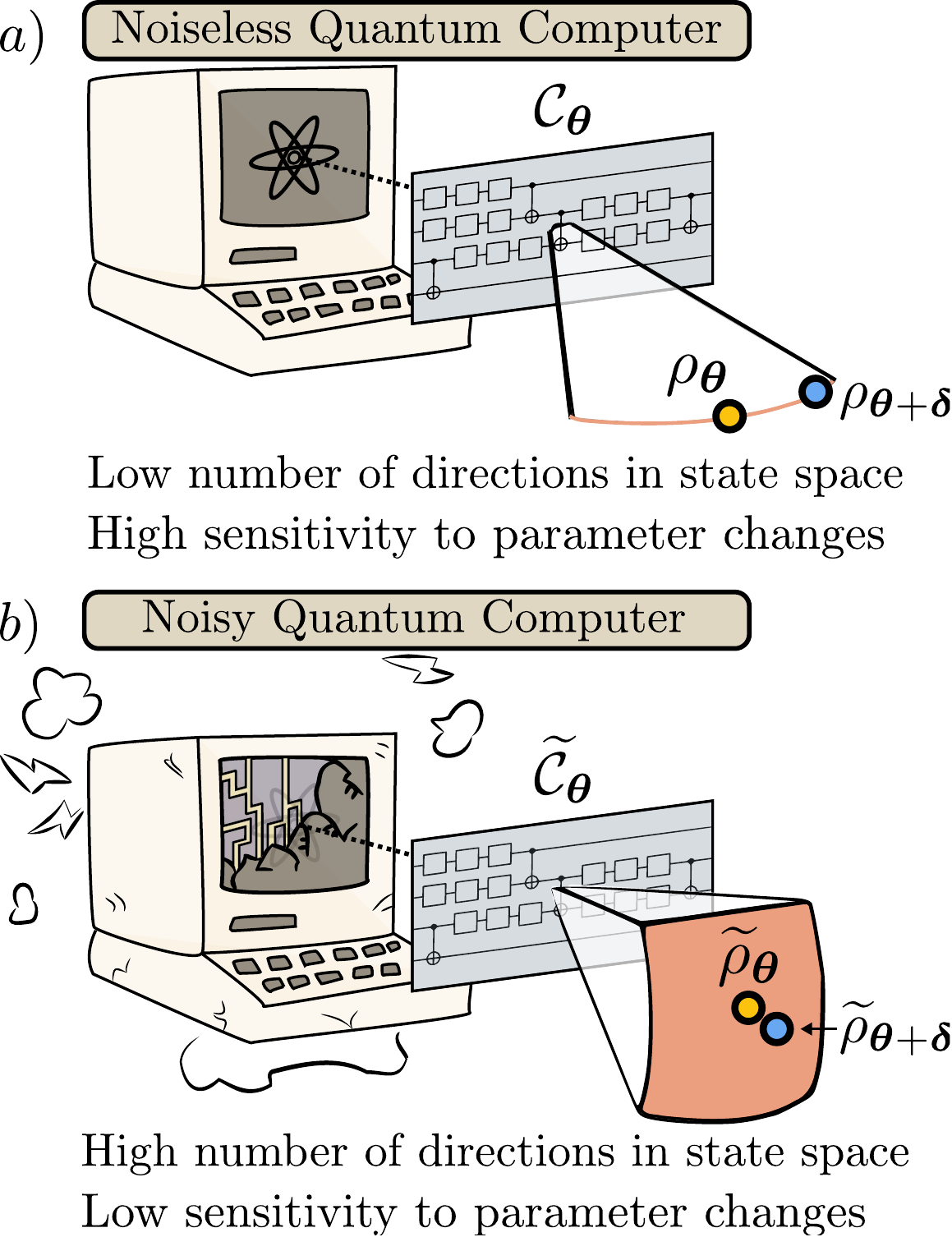}
    \caption{{\bf Schematic diagram of our main results.}  Consider the task of implementing a QNN, i.e., a parametrized unitary channel on a  quantum computer. As shown in~\cite{larocca2021theory}, the overparametrization phenomenon is defined as the QNN having enough parameters to explore all relevant directions in state space. a)  For certain ansatzes the QNN can be efficiently overparametrized with few parameters, as there only exists a small number of available directions in state space. Moreover, for (most) such directions, changes in the parameter values usually translate into changes in state space. 
    b) When the quantum device is faulty, quantum noise will act throughout the computation. In this work, we explore how hardware noise modifies the overparametrization phenomenon. Our results show that quantum noise can enable additional directions in state space. However, we also find that as the noise probability increases, the system becomes more and more insensitive to variations in the parameters. }
    \label{fig:main}   
\end{figure}

\section{Framework}

\subsection{Quantum Neural Networks}

In this work we consider a QML task where the goal is to train a model on a dataset $\SC=\{\rho^{(s)} \}_{s=1}^N$ consisting of $n$-qubit quantum states. We use $d$ to denote the dimension of the composite quantum system, i.e., $d=2^n$. The quantum model is parametrized through a QNN, which is a  unitary quantum channel $\CC_{\thv}$ acting on input states $\rho^{(s)}$ as $\CC_{\thv}(\rho^{(s)})=U(\thv)\rho^{(s)} U\ad(\thv)$. Here, $U(\thv)$ is taken to be of the form
\begin{equation}\label{eq:PSA_ansatz}
    U(\thv)=\prod_{m=1}^M U_m(\theta_m)\,, \quad U_m(\theta_m)= e^{-i \theta_{m}H_m}\,,
\end{equation}
where $H_m$ are traceless Hermitian operators taken from a set of generators $\GC$, and $\thv\in\mathbb{R}^M$ is a vector of trainable parameters. The previous allows us to express $\CC_{\thv}$ as a concatenation of $M$ unitary channels
\begin{equation}\label{eq:channels}
\CC_{\thv}=\CC^M_{\theta_M}\circ\cdots\circ\CC^1_{\theta_1}\,,
\end{equation}
with  $\CC^m_{\theta_m}(\rho^{(s)})=U_m(\theta_m)\rho^{(s)} U_m\ad(\theta_m)$. Thus, the output of the QNN is a parametrized state 
\begin{equation}   \label{eq:noiseless_state}
\rho^{(s)}_{\thv}=\CC_{\thv}(\rho^{(s)})=\CC^M_{\theta_M}\circ\cdots\circ\CC^1_{\theta_1}(\rho^{(s)})\,.
\end{equation}
 The variational parameters $\thv$ are trained by minimizing an appropriately chosen loss function $\LC(\thv)$ which we consider  to be of the form
\begin{equation}\label{eq:loss-function}
\LC(\thv)=\sum_{s=1}^N f_s\left(\Tr[\CC_{\thv}(\rho^{(s)})O_s]\right)\,,
\end{equation}
where $O_s$ and $f_s$ are respectively a (potentially) data-instance-dependent measurement, and a post-processing function.

While there are many aspects that define and distinguish a given QNN from another, we note that one of the most important is the choice of generators $\GC$ from which the QNN in Eq.~\eqref{eq:PSA_ansatz} is built. Once $\GC$ is determined, the next aspect that defines a QNN is its depth, or equivalently, the number of parameters $M$. In particular, one wants to choose  $\GC$ and $M$ such that there exist parameters values for which the task at hand is solved. While in this work we will not discuss how to appropriately choose $\GC$ (we instead refer the reader to~\cite{nguyen2022atheory,larocca2021diagnosing}), let us consider the effect of increasing the value of $M$. In a nutshell, adding more parameters to a QNN increases its expressibility (up to a certain point)~\cite{sim2019expressibility,holmes2021connecting,larocca2021diagnosing}, meaning that the QNN can generate a wider breadth of unitaries. From a practical stand-point, adding new parameters can potentially enable new directions in the state space~\footnote{By ``directions in state space'' we refer to elements of the tangent hyperplane defined at any point in state space~\cite{heydari2015geometric}}, and concomitantly in the loss functions landscape. This can improve the trainability of the model by removing spurious local minima, and increasing the dimension of the solution manifold~\cite{larocca2020exploiting,larocca2021theory}. In the following subsection we will see that the overparametrization phenomenon is indeed linked to the number of independent directions that are accessible in state space.

\subsection{Dynamical Lie algebra, quantum Fisher information, and overparametrization}

We will briefly recall here the main results in~\cite{larocca2021theory}. We will begin by defining the Dynamical Lie Algebra (DLA) of a QNN~\cite{dalessandro2010introduction,zeier2011symmetry}, which can be used to characterize the group of unitaries that it can be implemented~\cite{sim2019expressibility,holmes2021connecting,larocca2021diagnosing}. It follows that the DLA also determines the manifold of all reachable states by the QNN. This will allow us to interpret the overparametrization regime as that in which the QNN has enough parameters to explore all accessible directions in said manifold. In particular, we will show that the rank of the quantum Fisher information matrix can be used to detect the onset of overparametrization.

\begin{definition}[Dynamical Lie Algebra]\label{def:dynamical_lie_algebra} Given a set of Hermitian generators $\GC$, the dynamical Lie algebra $\liea$ is the subspace of operator space spanned by the repeated nested commutators of the elements in $i\GC$. That is
\begin{equation}
\liea={\rm span}_{\mathbb{R}}\left\langle i\GC \right\rangle_{Lie}\,,
\end{equation}
where $\left\langle i\GC \right\rangle_{Lie}$ denotes the Lie closure of $i\GC$.
\end{definition}

The DLA contains information about the ultimate expressiveness of the QNN, since the group of reachable unitaries obtained for any possible parameter values $\thv\in\mathbb{R}^M$ (for an arbitrary large number of parameters $M$) is obtained from the DLA via exponentiation, i.e., as  $\{ U(\thv)\}_{\thv}=\mathbb{G}=e^\liea \subseteq \SC\UC(d)$. We remark that $\mathbb{G}$ is known as the dynamical Lie group. Moreover, the manifold of states obtained from the action of the QNN on an input state $\rho$, given by 
$\{U\rho U\ad, U\in\mathbb{G}\}$, is known as the orbit of $\rho$ under $\mathbb{G}$.

From here, we can ask: \textit{By varying the parameters in the QNN, can we explore all accessible directions in the orbits of the input states, i.e., are we in the overparametrized regime?} To answer this question, let us assume for now that the dataset consist of a single parametrized pure state $\ket{\psi}$. One can study the action of the QNN on $\ket{\psi}$  via the Quantum Fisher Information Matrix (QFIM). To define the QFIM, start by considering a distance measure $\DC$ between two pure states. In particular, we take $\DC$ to be the infidelity, i.e., $\DC(\ket{\psi},\ket{\phi})=1-|\langle\psi|\phi\rangle|^2$. Then, given a set of parameters $\thv$ and an infinitesimal perturbation $\vec{\delta}$, an expansion to second-order of  $\DC$ between the quantum states $\ket{\psi(\thv)}=U(\thv)\ket{\psi}$ and $\ket{\psi(\thv+\vec{\delta})}=U(\thv+\vec{\delta})\ket{\psi}$ gives the Fubini-Study metric~\cite{cheng2010quantum,meyer2021fisher}, i.e., 
\begin{equation} \label{eq:pure-qfim-dis}
    \DC(\ket{\psi(\thv)},\ket{\psi(\thv+\vec{\delta})})=\frac{1}{2}\delv^T\cdot F(\ket{\psi(\thv)})\cdot\delv\,.
\end{equation}
Here, $F(\ket{\psi(\thv)})$ is the QFIM for the state $\ket{\psi(\thv)}$, an  $M\times M$ matrix whose elements are given by~\cite{liu2019quantum}
\begin{align}\label{eq:QFIM-elem}
\begin{split}
[F(\ket{\psi(\thv)})]_{ij}\!=\!4\Re[&\braket{\partial_i\psi(\thv)}{\partial_j\psi(\thv)}\\
&-\braket{\partial_i\psi(\thv)}{\psi(\thv)}\braket{\psi(\thv)}{\partial_j\psi(\thv)}]\,,
\end{split}
\end{align}
where $\ket{\partial_i\psi(\thv)}=\partial \ket{\psi(\thv)}/\partial\theta_i=\partial_i\ket{\psi(\thv)}$, for $\theta_i\in\thv$. 
As  shown in Fig.~\ref{fig:QFIM_directions}, the eigenvalues and eigenvectors of the QFIM provide valuable geometrical information regarding how changes in the parameters translate into changes in the state. Crucially, the rank of the QFIM quantifies the number of independent directions in state space that can be explored by making infinitesimal changes in $\thv$.

\begin{figure}[t!]
    \centering
    \includegraphics[width=1\linewidth]{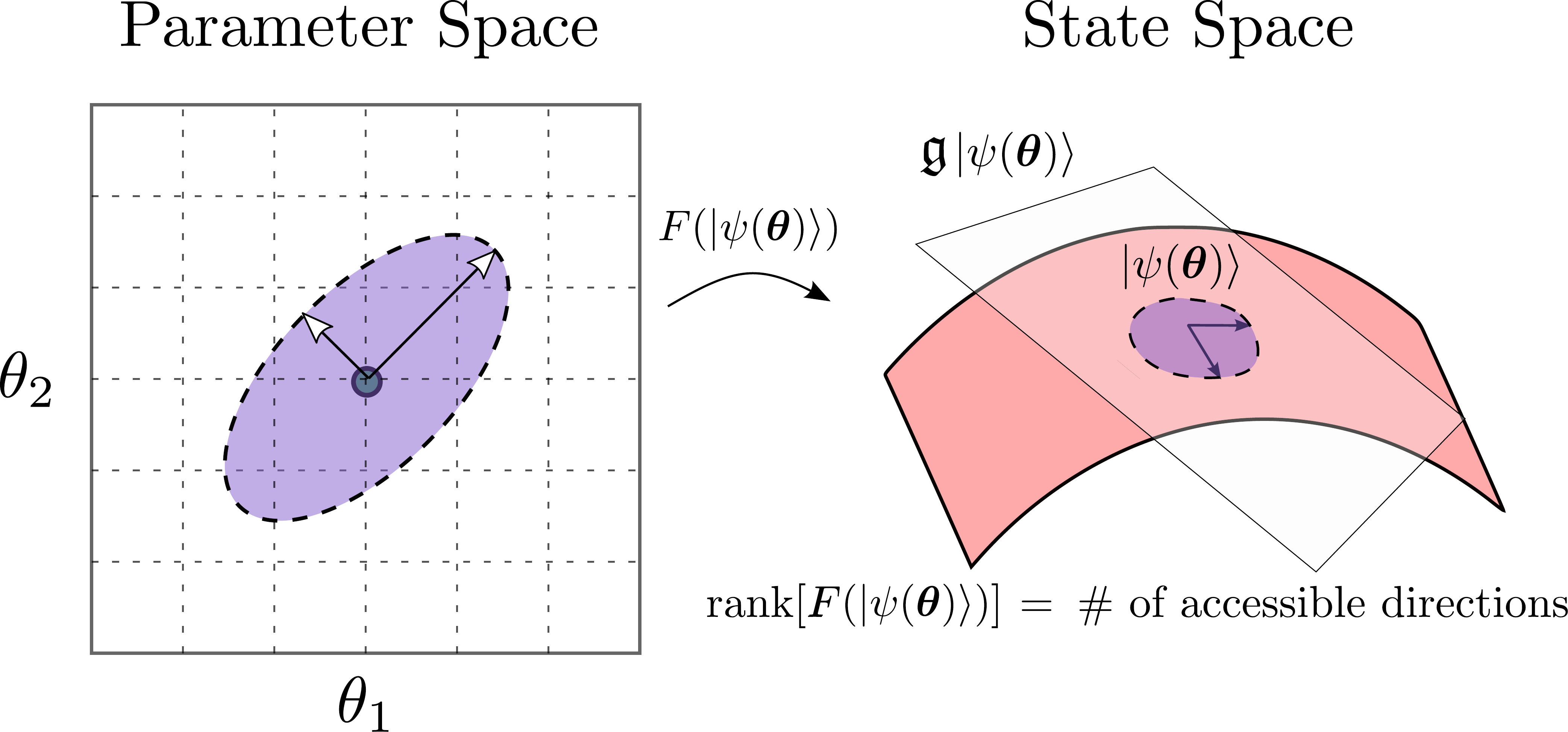}
    \caption{{\bf QFIM and directions in state space.} Let $U(\thv)\in e^{\mathfrak{g}}$ be a parametrized unitary and $\ket{\psi}$ some pure state, so that  $\ket{\psi(\thv)}=U(\thv)\ket{\psi}$. Here we schematically show that the eigenvalues and eigenvectors of the QFIM, $F(\ket{\psi(\thv)})$, inform how changes in parameter space translate into changes in state space. In particular, when modifying $\thv$ following  the  eigenvectors of the  QFIM, the state $\ket{\psi(\thv)}$ explores the corresponding available direction in the tangent space $\mathfrak{g}\ket{\psi(\thv)}$. Additionally, the magnitude of the QFIM eigenvalues determines the sensitivity of the state to a change along an eigenvector direction~\cite{meyer2021fisher}. As such, a large eigenvalue means that it is ``easy'' to nudge the state in state space, while small eigenvalues indicate that the state is insensitive to parameter changes in the direction of the associated eigenvector.        }
    \label{fig:QFIM_directions}   
\end{figure}

As such, one can determine if the QNN is overparametrized by checking if it has enough parameters so that the QFIM saturates its maximum achievable rank. 

\begin{definition}[Overparametrization]\label{def:overparametrization}
    A QNN is said to be overparametrized if the number of parameters $M$ is such that the QFIM saturates its achievable rank $R$ at least in one point of the loss landscape. That is, if increasing the number of parameters past some minimal (critical) value $M_c$ does not further increase the rank of the QFIM, i.e.,
    \begin{equation}\label{eq:Rmu}
        \max_{M\geq M_c,\thv}\rank[F(\ket{\psi(\thv)})] = R\,.
    \end{equation}
\end{definition}

The main result of Ref.~\cite{larocca2021theory} is that $M_c$ is directly linked to the dimension of $\liea$ (the circuit's DLA). In particular, the rank of the QFIM is upper bounded by $\dim(\liea)$. Hence, one can potentially reach the overparametrization regime if the QNN has $\sim\dim(\liea)$ parameters. Clearly, if $\dim(\liea)\in\Omega(\rm{exp}(n))$ (as is the case of controllable unitaries) then one cannot efficiently overparametrize the QNN. More interesting, however, are the cases when $\dim(\liea)\in\OC(\poly(n))$, such as those arising in~\cite{larocca2021diagnosing,schatzki2022theoretical}. We remark that while we have defined overparametrization as the regime where the number of parameters is such that the QFIM saturates its achievable rank in at least one point in the landscape, in practice one finds that the rank saturates simultaneously throughout most of the landscape (which is what brings about a computational phase transition)~\cite{larocca2021theory}.

To finish, we note that in Ref.~\cite{larocca2021theory} it was also shown that Definition~\ref{def:overparametrization} has operational meaning in terms of the capacity of the QNN~\cite{haug2021capacity,abbas2020power}. We recall that the capacity (or power) of a QNN is used to quantify the breadth of functions
that it can capture~\cite{coles2021seeking}. For instance, let us consider the capacity measure of~\cite{haug2021capacity}, which defines the effective quantum dimension of a QNN as
\begin{equation}\label{eq:eff_dim_1}   D_1(\thv)=\mathbb{E}\left[\sum_{m=1}^M\ZC(\lambda^{m}(\thv))\right]\,.
\end{equation}
Here $\lambda^{m}(\thv)$ are the eigenvalues of the QFIM for the state $\ket{\psi(\thv)}$, and $\ZC(x)$ is a function such that $\ZC(x)=0$ for $x=0$, and $\ZC(x)=1$ for $x\neq0$. Moreover, the expectation value is taken over the probability distribution  from which states are sampled from $\SC$.  It is straightforward to see that for a single-state dataset, $D_1(\thv)=\rank[F(\ket{\psi(\thv)})]$. An alternative definition for the capacity of a QNN can be found in~\cite{abbas2020power}. In the limit of large datasets, i.e., when $|\SC|\rightarrow\infty$, the effective quantum dimension of~\cite{abbas2020power} converges to
\begin{equation}\label{eq:eff_dim_2}
    D_2=\max_{\thv}\left(\rank\left[I(\ket{\psi(\thv)})\right]\right)\,,
\end{equation}
where $I(\ket{\psi(\thv)})$ is the classical Fisher Information matrix, defined as 
\small
\begin{equation}    I(\ket{\psi(\thv)})=\mathbb{E}\left[\frac{\partial\log(p(\ket{\psi},y;\thv))}{\partial\thv}\frac{\partial\log(p(\ket{\psi},y;\thv))}{\partial\thv}^T\right].
\end{equation}
\normalsize 
Here, $p(\ket{\psi},y;\thv)$, describes the joint relationship between an input $\ket{\psi}$ and an output $y$ of the QNN. In addition, the expectation value is taken over the probability distribution that samples input states from the dataset. As shown in~\cite{larocca2021theory} the model's capacity, as quantified by the effective dimensions of Eqs.~\eqref{eq:eff_dim_1} or~\eqref{eq:eff_dim_2}, is upper bounded as
\begin{equation}
    D_1(\thv)\leq \dim(\liea),\quad D_2\leq \dim(\liea)\,.
\end{equation}
Moreover, one can show that when the QNN is overparametrized, $D_1(\thv)$ achieves its maximum value. This shows that overparametrizing a QNN is equivalent to saturating its capacity.

\subsection{Quantum noise preliminaries}

Quantum noise refers to the uncontrolled errors that occur when implementing a QNN on quantum hardware. Such errors may arise from a wide variety of sources, such as imperfections when implementing gates or when performing measurements, undesired qubit-qubit couplings or unwanted interactions between the qubits and their environment.

In this work, we model the action of  the hardware noise  present throughout 
a QNN  by considering that noise channels $\NC_m$ act before and after each unitary $U_m(\theta_m)$  (see Fig.~\ref{fig:circs}). Here, we recall the definition of a unital Pauli channel.
\begin{definition}[Unital Pauli channel] \label{def:pauli_channel}
A unital Pauli channel is a CPTP map $\NC$ whose action on an operator $\rho$ is given by 
\begin{equation} \label{eq:pauli_channel_prob}
    \NC(\rho) = \sum_{\vec{\alpha}\vec{\beta}}  p_{\vec{\alpha}\vec{\beta}} X^{\vec{\alpha}}Z^{\vec{\beta}} \rho Z^{\vec{\beta}} X^{\vec{\alpha}}\,,
\end{equation}
where $\{p_{\vec{\alpha}\vec{\beta}}\}$ is a probability distribution (i.e.,  $p_{\vec{\alpha}\vec{\beta}}\geq 0$ and $\sum_{\vec{\alpha}\vec{\beta}} p_{\vec{\alpha}\vec{\beta}} =1$), and $X^{\vec{\alpha}} Z^{\vec{\beta}} \coloneqq X^{\alpha_1}Z^{\beta_1}\otimes ... \otimes X^{\alpha_n}Z^{\beta_n}$, where $\alpha_1, \dots, \alpha_n, \beta_1, \dots, \beta_n \in \{0, 1\}$.
\end{definition}

In other words, a unital (identity preserving) Pauli channel consists of Pauli operators applied randomly according to a certain probability distribution. It is easy to see that it is diagonal in the Pauli basis. That is, its action maps a Pauli operator  $X^{\vec{\alpha}'} Z^{\vec{\beta}'}$  onto itself as $\NC( X^{\vec{\alpha}'} Z^{\vec{\beta}'} ) = c_{\vec{\alpha}'\vec{\beta}'} X^{\vec{\alpha}'} Z^{\vec{\beta}'}$, where $c_{\vec{0}\vec{0}}=1$ and $-1\leq c_{\vec{\alpha}'\vec{\beta}'}\leq 1$ for all $\vec{\alpha}'$ and $\vec{\beta}'$. Indeed,
\begin{equation} \begin{split}
    \NC(X^{\vec{\alpha}'} Z^{\vec{\beta}'}) &= \sum_{\vec{\alpha}\vec{\beta}} p_{\vec{\alpha}\vec{\beta}} X^{\vec{\alpha}}Z^{\vec{\beta}} X^{\vec{\alpha}'} Z^{\vec{\beta}'} Z^{\vec{\beta}} X^{\vec{\alpha}}  \\ & = X^{\vec{\alpha}'} Z^{\vec{\beta}'} \underbrace{\sum_{\vec{\alpha}\vec{\beta}} (-1)^{\vec{\alpha'}\cdot\vec{\beta}}  (-1)^{\vec{\alpha}\cdot\vec{\beta'}} \,p_{\vec{\alpha}\vec{\beta}}}_{c_{\vec{\alpha'}\vec{\beta'}}} \,,\end{split}
\end{equation}
where we used the following properties,
\begin{equation}\label{eq-pauli_prop}
[X^{\vec{\alpha}}, X^{\vec{\alpha}'}]=0,\quad\![Z^{\vec{\beta}},Z^{\vec{\beta}'}]=0,\quad \!X^{\vec{\alpha}}Z^{\vec{\beta}}=(-1)^{\vec{\alpha}\cdot\vec{\beta}}Z^{\vec{\beta}}X^{\vec{\alpha}},\!\nonumber
\end{equation}
together with the fact that the square of a Pauli operator is  equal to the identity. Note that  $c_{\vec{0}\vec{0}}=1$ implies that a unital Pauli noise channel maps the identity operator onto itself, which is a necessary and sufficient condition for a diagonal superoperator to be trace preserving. In what follows we will assume that $c_{\vec{\alpha'}\vec{\beta'}}\in(-1,1)$ for all $\vec{\alpha'}$ and $\vec{\beta'}$ (this is necessary for Lemma~\ref{lem-renyi} in Appendix~\ref{ap:lemm} to hold, i.e., for the identity operator to be the only fixed point of the noisy channel).

Pauli unital noise includes, as a special case,
\textit{local depolarizing noise}, which acts on each qubit $j\in[1,n]$ as~\cite{wilde2013quantum}
\begin{align}\label{eq:local_depolarizing}
    \NC^{Depol}_j(\rho)&=\left(1-\frac{3p}{4}\right)\rho+\frac{p}{4}\left(X_j\rho X_j+Y_j\rho Y_j+Z_j\rho Z_j\right)\,,\nonumber\\
    &=(1-p)\rho + p\frac{\id_j\otimes\Tr_{j}[\rho]}{2}\,.
\end{align}
Here, $ 0< p \leq 1$ denotes the probability of depolarization, and $X_j$, $Y_j$ and $Z_j$  are  Pauli operators acting on the $j$-th qubit. Moreover, $\Tr_{j}$ indicates the partial trace over qubit $j$. Similarly, we can construct an $n$-qubit channel consisting of a local depolarizing channel acting on each qubit as 
\begin{equation}\label{eq:local-depol-noise}
\NC^{Depol}_{loc}(\rho)=\bigotimes_{j=1}^n \NC^{Depol}_j(\rho)\,,
\end{equation}
or the \textit{global depolarizing channel}, whose action is
\begin{equation} \label{eq:global_dep}
    \NC^{Depol}(\rho)=(1-p)\rho + p\frac{\id}{d}\,,
\end{equation}
where $0<p\leq 1$. 
Other examples of Pauli noise channels include
bit- and phase-flip channels, 
as well as $T_2$ processes (i.e., the dephasing channel is a unital Pauli channel). 

\begin{figure}[t!]
    \centering
    \includegraphics[width=.8\linewidth]{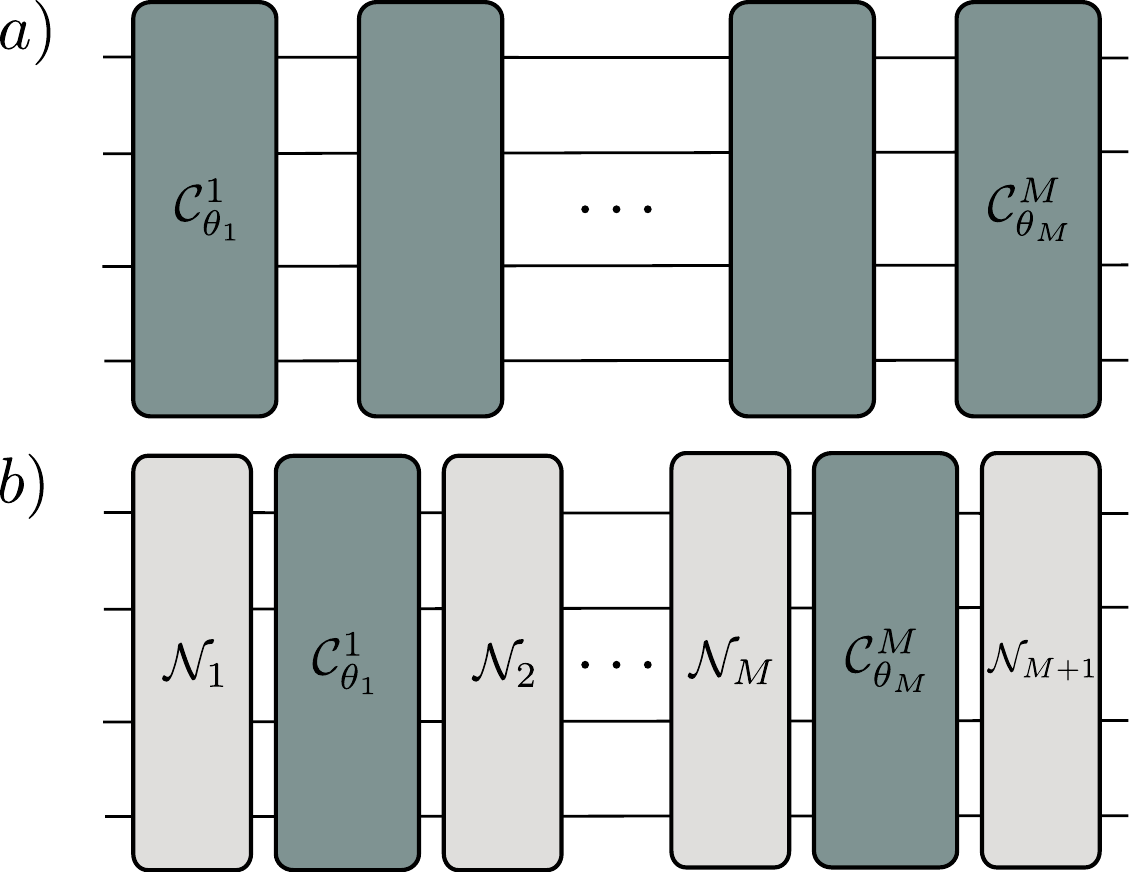}
    \caption{{\bf Noiseless and noisy quantum circuits.} a) Noiseless quantum circuit consisting of parametrized unitary channels $\CC^m_{\theta_m}$. b) Noisy quantum circuit where unital Pauli noise channels $\NC_m$ are interleaved with the unitary channels.}
    \label{fig:circs}   
\end{figure}

As shown in Fig.~\ref{fig:circs}, in the presence of quantum noise the action of the QNN is modeled by noise channels interleaved with the unitary channels. Hence, the output of the noisy QNN is given by
\begin{equation}\label{eq:noisy-channel}
   \widetilde{\rho}_{\thv}= \widetilde{\CC}_{\thv}(\rho)=\NC_{M+1}\circ\CC^M_{\theta_M}\circ\NC_{M}\circ\cdots\circ\NC_2\circ\CC^1_{\theta_1}\circ\NC_1(\rho)\,,
\end{equation}
for some (potentially layer-dependent) noise channels $\NC_m$, with $m=1,\ldots,M+1$.

\subsection{Mixed-state quantum Fisher information matrix}

As previously discussed, in the presence of noise the quantum states  evolving through the circuit become mixed, and we must extend the formula of the QFIM in Eq.~\eqref{eq:QFIM-elem} to account for this. Following the same program that led to the QFIM for pure states, we can define a mixed-state QFIM as an expansion of the \textit{Bures distance}, which is a measure of distinguishability between mixed states. The Bures distance  is defined as
\begin{equation}\label{eq:bures}
    \BC(\rho,\sigma)=2\left(1-\sqrt{\FC(\rho,\sigma)}\right)\,,
\end{equation} 
where $\FC(\rho,\sigma)$ is the Ulhmann fidelity~\cite{uhlmann1976transition}
\begin{equation}
\FC(\rho,\sigma)=\left(\Tr[\sqrt{\sqrt{\rho}\sigma\sqrt{\rho}}]\right)^2\,.
\end{equation}

Let $\rho_{\thv}$ be a parametrized mixed state. And let its  spectral decomposition be
\begin{equation}
    \rho_{\thv}=\sum_{\mu=1}^d r_\mu \dya{r_\mu}\,,
\end{equation}
where $\{r_\mu\}_{\mu=1}^d$ are the eigenvalues of $\rho_{\thv}$ (such that $r_\mu\geq 0$ for all $\mu$, and $\sum_\mu r_\mu=1$), and $\{\ket{r_\mu}\}_{\mu=1}^d$ are the associated eigenvectors\footnote{For simplicity, we omit the $\thv$ dependency in $r_\mu$ and $\ket{r_\mu}$.}. Then, a second order expansion of the Bures distance between ${\rho}_{\thv}$ and ${\rho}_{\thv+\vec{\delta}}$  leads to the mixed state QFIM (which reduces to Eq.~\eqref{eq:QFIM-elem} when $\rho_{\thv}$ is pure), whose entries are~\cite{liu2014fidelity,meyer2021fisher}  
\small
\begin{align}\label{eq:QFIM-mixed-elem-si}
[F(\rho_{\thv})]_{ij}&=
\sum_{\substack{\mu,\nu\\ r_{\mu}+ r_{\nu}\neq0}}
\frac{2\Re[\bra{r_{\mu}}\partial_i\rho_{\thv}\ket{r_{\nu}}\bra{r_{\nu}}\partial_j\rho_{\thv}\ket{r_{\mu}}]}{r_\mu+r_\nu}\\
 &=\sum_{\substack{\mu\\ r_\mu\neq 0}}\left(\frac{(\partial_i{r}_\mu)(\partial_j{r}_\mu)}{{r}_\mu}+4{r}_\mu\Re\left[\langle\partial_i{r}_\mu|\partial_j{r}_\mu\rangle\right]\right)\nonumber\\
&\quad -\sum_{\substack{\mu,\nu\\ r_\mu+ {r}_\nu\neq0}}\frac{8{r}_\mu{r}_\nu}{{r}_\mu+{r}_\nu}\Re\left[\langle\partial_i{r}_\mu|{r}_\nu\rangle\langle {r}_\nu|\partial_j{r}_\mu\rangle\right]\,.
\end{align}
\normalsize

Here we recall a few properties of the QFIM which will be used below, and we refer the reader to~\cite{liu2014fidelity} for their  proof.
\begin{enumerate}
    \item $F$ is a symmetric matrix: $F^T=F$.
    \item $F$ is positive semi-definite: $F\geq 0$.
    \item  $F$ is convex: For any pair of states $\rho_{\thv}$ and $\sigma_{\thv}$ and for $0\leq q \leq 1$ we have $F(q\rho_{\thv}+(1-q)\sigma_{\thv})\leq qF(\rho_{\thv})+(1-q)F(\sigma_{\thv})$.
    \item $F$ is invariant under unitary transformations: $F(U\rho_{\thv}U\ad)=F(\rho_{\thv})$ for any $U\in\UC(d)$.
    \item $F$ is non-increasing under quantum channels: If $\Phi$ is a quantum channel, then $F(\Phi(\rho_{\thv}))\leq F(\rho_{\thv})$.
\end{enumerate}

\section{Results}
\label{sec:main_results}

The previous section reviewed the results of Ref.~\cite{larocca2021theory}, which analyzed the overparametrization phenomenon when no hardware noise is present. However, in a realistic scenario where the QNN is implemented on a near-term quantum device~\cite{preskill2018quantum} we can expect that quantum noise will act throughout the circuit. Therefore, in what follows we set out to study how the results of Ref.~\cite{larocca2021theory} change when noise is considered. For simplicity, we will study the case when the dataset is composed of a single mixed state $\SC=\{\rho\}$. The extension of our results to  multi-state datasets is straightforward, as one simply needs to follow  the approach taken in~\cite{larocca2021theory}.

\begin{figure*}[t!]
    \centering
    \includegraphics[width=.6\linewidth]{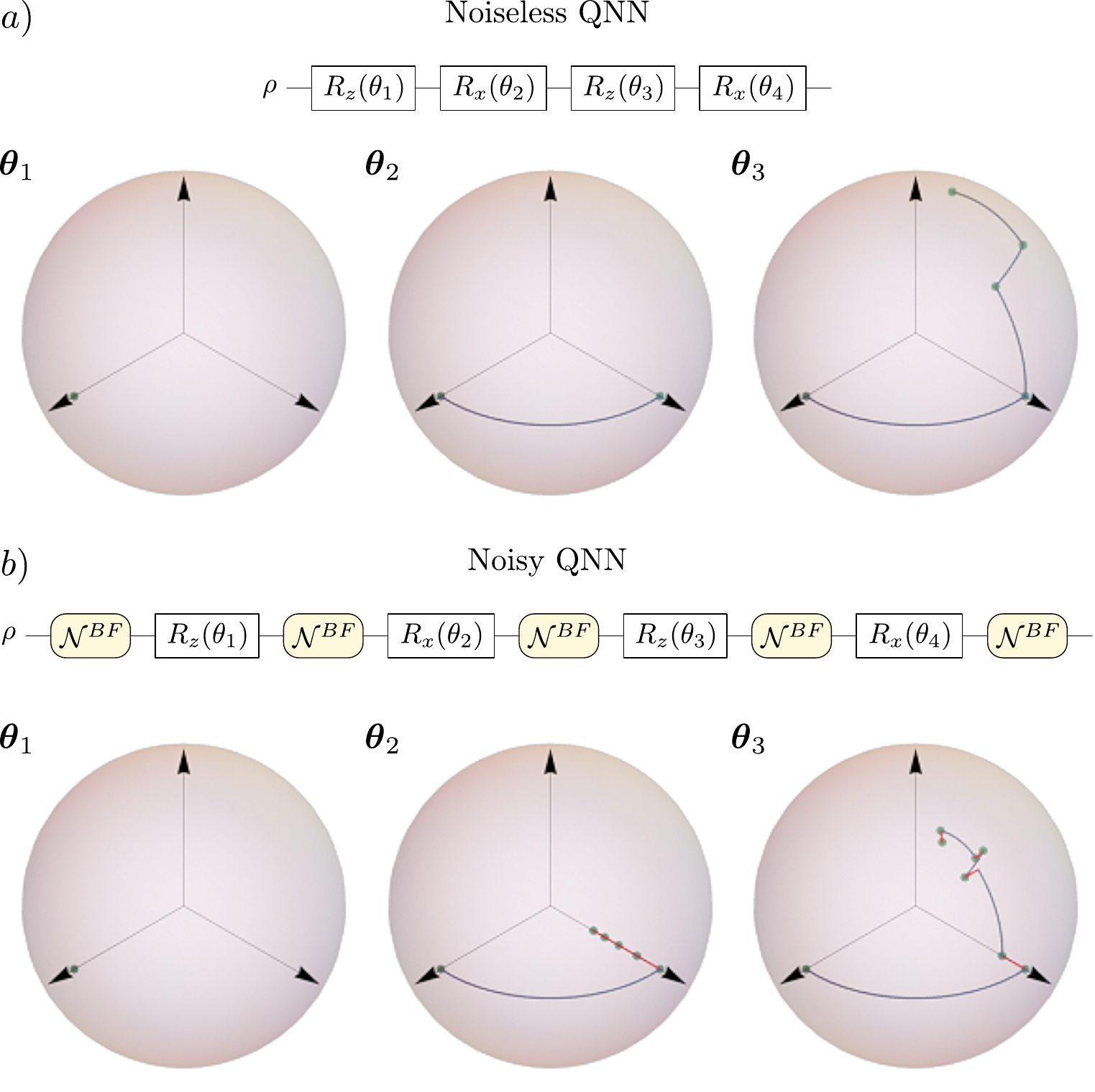}
    \caption{{\bf Single-qubit toy model examples.} a) We consider the case where the single qubit state of Eq.~\eqref{eq:si-state} is sent through a noiseless QNN with four parameters as in Eq.~\eqref{eq:si_qnn}. We plot in the Bloch sphere the three trajectories defined by $\thv_1$, $\thv_2$ and $\thv_3$. b) We consider the case where the single qubit state of Eq.~\eqref{eq:si-state} is sent through a noisy QNN with four parameters, as in Eq.~\eqref{eq:si-noisy-qnn}. Here,  bit-flip noise channels act before and after every gate with probability $p=0.1$. We plot in the Bloch sphere the three trajectories defined by $\thv_1$, $\thv_2$ and $\thv_3$. The action of the unitary gates is marked in blue, whereas the action of the noise channels is marked in red.}
    \label{fig:noise}
\end{figure*}

\subsection{Single-qubit toy model}
\label{sec:single-qubit}

We start with a simple toy model that will help us gather intuition on the effects that noise may have on the rank and the eigenvalues of the QFIM, and hence on the QNNs' overparametrization. As we will show, we can expect that  presence of quantum noise will generally: i) Increase the rank of the QFIM, and ii) Decrease the overall magnitude of the QFIM eigenvalues.  To illustrate these two phenomena, we consider a simple single-qubit model undergoing bit-flip noise. The setup is as follows. First, we initialize the state of the single qubit to
\begin{equation}\label{eq:si-state}
    \rho=0.9\dya{+}+0.1\frac{\id}{2}\,.
\end{equation}
We choose a full rank state to avoid issues in the QFIM (namely, discontinuities in its entries) arising from a change in the rank of the state~\cite{vsafranek2017discontinuities,seveso2019discontinuity}. Then, this state is sent through a circuit composed of four single qubit rotations
\begin{equation}
    U(\thv)=e^{-i \theta_4 X/2}e^{-i \theta_3 Z/2}e^{-i \theta_2 X/2}e^{-i \theta_1 Z/2}\,.
\end{equation}
This setup is depicted in Fig.~\ref{fig:noise}(a, top). In channel notation, this QNN is expressed as the concatenation of four unitary channels
\begin{equation}\label{eq:si_qnn}
   \CC_{\thv}= \CC^X_{\theta_4}\circ\CC^Z_{\theta_3}\circ\CC^X_{\theta_2}\circ\CC^Z_{\theta_1}\,,
\end{equation}
where $\CC^Z_{\theta}(\rho)=e^{-i \theta Z/2}\rho e^{i \theta Z/2}$, and analogously for $\CC^X_{\theta}(\rho)$.

The generators of the QNN are the Pauli matrices $\GC=\{X,Z\}$, and it is straightforward to check that the DLA is simply $\liea=\spn\{iX,iY,iZ\}\iso\mathfrak{su}(2)$, meaning that the QNN is \textit{universal} or \textit{controllable}~\cite{dalessandro2010introduction,larocca2021diagnosing}. Moreover, we can see that the maximum possible rank of the QFIM is $\max_{\thv}\left(\rank\left[F(\rho_{\thv})\right]\right)=2$, as the state lives on a two-dimensional shell inside of the Bloch sphere. As such, it is clear that the QNN is already overparametrized, since the maximum attainable rank of the QFIM is smaller than the number of parameters. 
To exemplify how noise affects the QNN, we evaluate the QFIM at three different sets of parameter values,
\begin{enumerate}
    \item $\thv_1=\{0,0,0,0\}$, leading to $\rank\left[F(\rho_{\thv_1})\right]=1$,
    \item $\thv_2=\{\frac{\pi}{2},0,0,0\}$, leading to $\rank\left[F(\rho_{\thv_2})\right]=2$,
    \item $\thv_3=\{\frac{\pi}{2},\frac{\pi}{4},\frac{\pi}{4},\frac{\pi}{4}\}$, leading to $\rank\left[F(\rho_{\thv_3})\right]=2$.
\end{enumerate}
The (noiseless) trajectories corresponding to these choices are presented in Fig.~\ref{fig:noise}(a, bottom). While the rank of the QFIM is indeed saturated at $\thv_2$ and $\thv_3$, for $\thv_1$ we have $\rank\left[F(\rho_{\thv_1})\right]=1$ (this follows from $R_z(0)\rho R_z(0)^\dagger$ and $R_z(0)R_x(0)R_z(0)\rho R_z(0)^\dagger R_x(0)^\dagger R_z(0)^\dagger$ being eigenstates of $R_x$). We have thus added this example to showcase the important role that the interplay between the initial state and the QNN parameters has in determining the rank of the QFIM.

\begin{figure}[t!]
    \centering
    \includegraphics[width=1\linewidth]{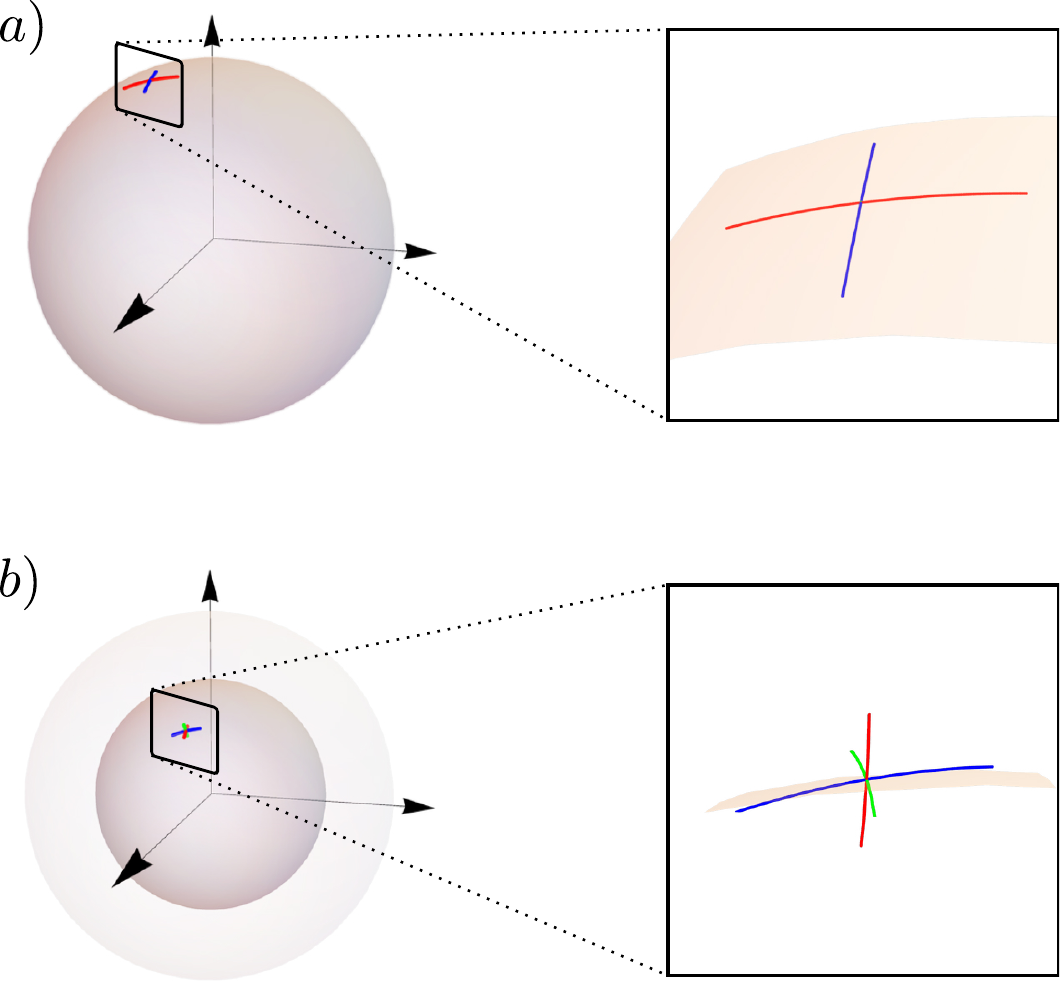}
    \caption{{\bf State space trajectories following perturbations along QFIM eigendirections.} a) We consider that the single qubit state of Eq.~\eqref{eq:si-state} is sent through a noiseless QNN as in Eq.~\eqref{eq:si_qnn}, with parameters $\thv_3$. Here we show within the Bloch sphere how the state $\rho_{\thv}$ changes when the parameters are varied following the directions given by  three eigenvectors of the QFIM $F(\rho_{\thv})$. Two such directions are associated with the two non-zero eigenvalues (blue and red curves) and with a zero eigenvalue (green, non-visible, curve). b) We consider that the single qubit state of Eq.~\eqref{eq:si-state} is sent through a noisy QNN as in Eq.~\eqref{eq:si-noisy-qnn}, with parameters $\thv_3$. Here we show within the Bloch sphere how the state $\rho_{\thv}$ changes when the parameters are varied following the directions given by  the three eigenvectors of the QFIM $F(\widetilde{\rho}_{\thv})$ with associated non-zero  eigenvalues (blue, red, and green curves).    }
    \label{fig:directions}
\end{figure}

Now, let us consider the case where bit-flip noise channels act before and after every unitary gate in the circuit (see Fig.~\ref{fig:noise}(b, top) for a schematic portrayal of the setup). For convenience, we recall that the bit-flip channel is a special case of Pauli noise of the form
\begin{equation}
    \NC^{BF}(\rho)=(1-p)\rho + pX\rho X\,, \quad \quad 0< p \leq 1\,,
\end{equation}
such that the noisy QNN channel becomes 
\begin{equation}\label{eq:si-noisy-qnn}
   \widetilde{\CC}_{\thv}= \NC^{BF}\circ\CC^X_{\theta_4}\circ\NC^{BF}\circ\CC^Z_{\theta_3}\circ\NC^{BF}\circ\CC^X_{\theta_2}\circ\NC^{BF}\circ\CC^Z_{\theta_1}\circ\NC^{BF}\,.
\end{equation}
A direct evaluation of the QFIM rank at the sets of parameter values previously considered reveals that 
\begin{enumerate}
    \item $\thv_1=\{0,0,0,0\}$, leads to $\rank\left[F(\rho_{\thv_1})\right]=1$,
    \item $\thv_2=\{\frac{\pi}{2},0,0,0\}$, leads to $\rank\left[F(\rho_{\thv_2})\right]=2$,
    \item $\thv_3=\{\frac{\pi}{2},\frac{\pi}{4},\frac{\pi}{4},\frac{\pi}{4}\}$, leads to $\rank\left[F(\rho_{\thv_3
    })\right]=3$.
\end{enumerate}
The trajectories defined by these rotations are presented in Fig.~\ref{fig:noise}(b, bottom), for a value of $p=0.1$. Here we can see that for $\vec{\theta}_1$ and $\thv_2$ the rank of the QFIM is not increased. While it is obvious that for $\thv_1$ the noise does not change the output state of the QNN (as $\rho$ is a fixed point of the noise model), for $\thv_2$ the noise channels do change the output state of the QNN. Notably, we can see that all the noise channels are effectively applied at the end of the parametrized evolution in both cases. As we will show below, this implies that they cannot change the rank of the QFIM (see Theorem~\ref{theo:noise-end}). Finally, for $\thv_3$ the noise does increases the rank of the QFIM from two to three. Here, the rank of the QFIM is maximal, which follows from  the state evolving in the three-dimensional Bloch sphere. 

\begin{figure}[t!]
    \centering
    \includegraphics[width=.7\linewidth]{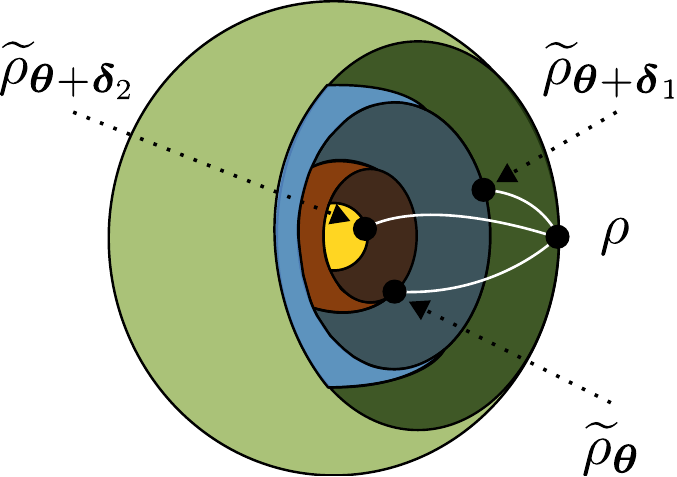}
    \caption{{\bf Purity and new directions in state-space.} Here, $\rho$ is a single-qubit  input state to a noisy QNN. Then, let $\rho_{\thv}$, $\rho_{\thv+\vec{\delta}_1}$ and $\rho_{\thv+\vec{\delta}_2}$ be the output states when the QNN parameters are $\thv$, $\thv+\vec{\delta}_1$ and $\thv+\vec{\delta}_2$, respectively. As schematically shown,  the parameters $\thv+\vec{\delta}_1$ ($\thv+\vec{\delta}_2$)  lead to the output state  $\rho_{\thv+\vec{\delta}_1}$ ($\rho_{\thv+\vec{\delta}_2}$) having more (less) purity than $\rho_{\thv}$, as the output state is farther (closer) to the center of the Bloch sphere. Note that in all cases, the output states $\rho_{\thv}$, $\rho_{\thv+\vec{\delta}_1}$ and $\rho_{\thv+\vec{\delta}_2}$  are less pure than the input state $\rho$ due to the presence of noise.  }
    \label{fig:models}   
\end{figure}

The fact that for $\thv_3$ the rank of the QFIM is increased indicates that the presence of noise enables a new direction in state space.   As shown in Fig.~\ref{fig:directions}(a), in the absence of noise (and hence when the rank of the QFIM is two), there are only two available directions in state space. These directions are depicted as blue and red lines corresponding to the trajectories followed by the state when the parameters are changed along the  directions dictated by the eigenvectors of the QFIM with non-zero eigenvalues. Since the channel is unitary, these trajectories lie on the surface of a (fixed-purity) shell of the Bloch sphere. We have also verified that, as expected, the state remains unchanged when the parameters are varied along a direction corresponding to an eigenvector associated to a null eigenvalue of the QFIM (although the previous cannot be visualized in the plot because the initial and final state of the evolution are the same). 

On the contrary, as shown in Fig.~\ref{fig:directions}(b), when noise acts throughout the circuit (and hence when the rank of the QFIM is three), there are three available directions in state space. Here,  the red, blue and green curves correspond to the trajectories that the state follows when changing the parameters along the directions given by the three eigenvectors of the QFIM with associated non-zero eigenvalue. Crucially, we can now see that there exists a direction (the blue curve) that preserves the purity of the quantum state. The other two directions, however, can both increase and decrease the purity of the output state. This is  evidenced from the fact that the trajectories in the state space move inwards and outwards from the fixed-purity shell in the Bloch sphere.

We find it important to remark that while some of the directions in state space can change the purity of the state $\widetilde{\rho}_{\thv}$, this does not imply that the QNN is purifying the state. As shown in Fig.~\ref{fig:directions}(b), by perturbing the parameters $\thv$ along a direction $\vec{\delta}$  one can move the \textit{final state} inwards or outwards in the Bloch sphere. That is, we can decrease or increase the purity of $\widetilde{\rho}_{\thv+\vec{\delta}}$ \textit{with respect to that of} $\widetilde{\rho}_{\thv}$, see Fig.~\ref{fig:models}. However, this does not imply that $\widetilde{\rho}_{\thv+\vec{\delta}}$ has less entropy than the initial state $\rho$. In other words, changing the variational parameters by $\vec{\delta}$ implies preparing again the initial state $\rho$ and applying $\widetilde{\CC}_{\thv+\vec{\delta}}$, not evolving from $\widetilde{\rho}_{\thv}$ to   $\widetilde{\rho}_{\thv+\vec{\delta}}$. Physically, we can interpret the previous as saying  that the state evolving under the noisy QNN $\widetilde{\CC}_{\thv+\vec{\delta}}$ is less sensitive to noise than that evolving under $\widetilde{\CC}_{\thv}$, and hence its purity gets less degraded by noise. 

\begin{figure}[t!]
    \centering
    \includegraphics[width=.48\textwidth]{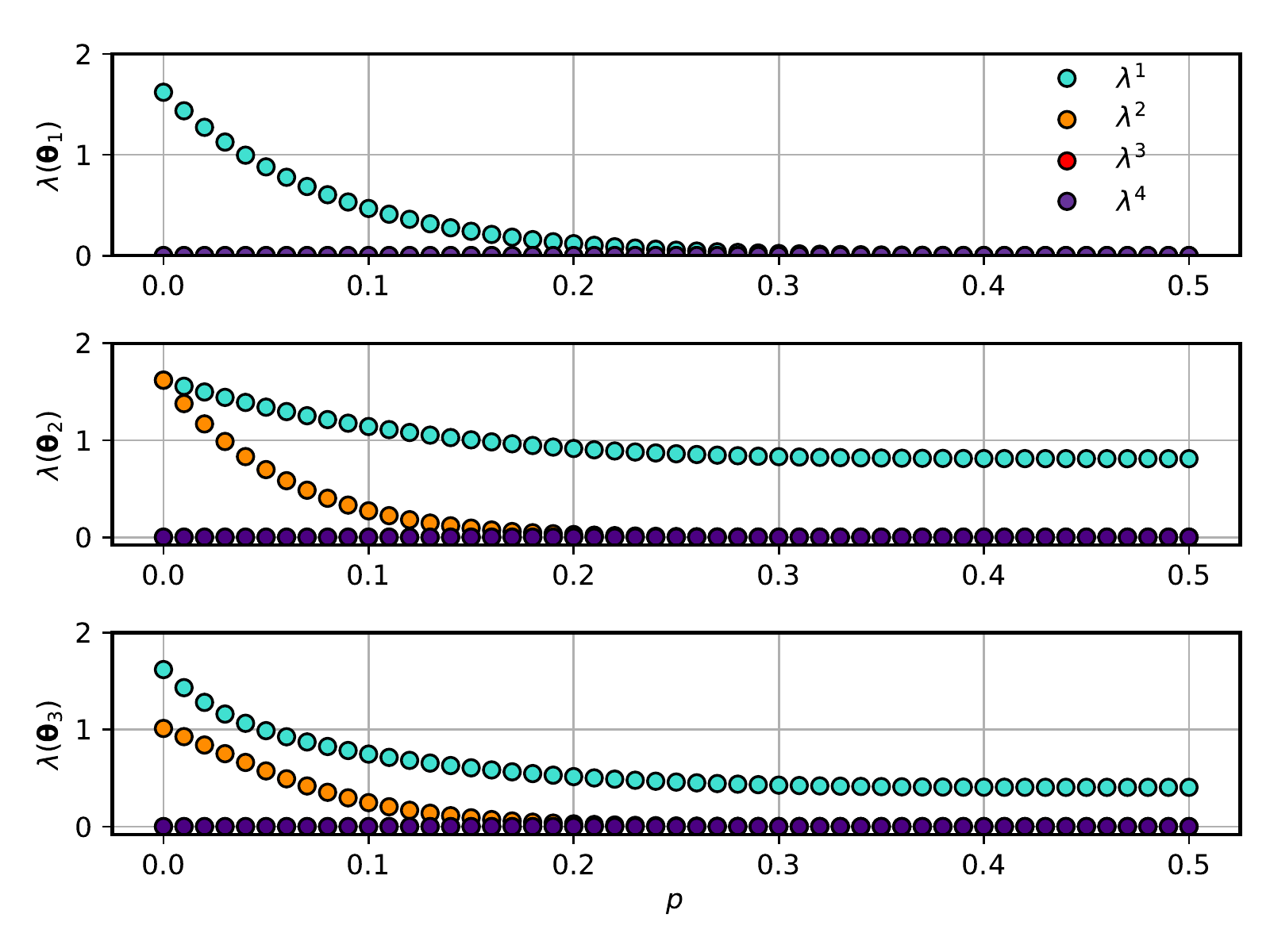}
    \caption{{\bf Eigenvalues of the QFIM versus noise levels.} Here we consider the case where  the noisy QNN from Eq.~\eqref{eq:si-noisy-qnn} acts on the initial single-qubit state of Eq.~\eqref{eq:si-state}. We plot the magnitude of the eigenvalues of the QFIM for the single-qubit toy model  versus the probability of bit-flip error $p$. The top, middle and bottom panels respectively correspond to the parameter values $\thv_1, \thv_2, \thv_3$.}
    \label{fig:eigvals_single_qubit}
\end{figure}

Next, let us evaluate how noise affects the  magnitude of the QFIM eigenvalues. In Fig.~\ref{fig:eigvals_single_qubit} we plot the eigenvalues of the QFIM versus the probability of a bit-flip error $p$. In this case, it is manifest that for all the parameter values considered above, the magnitude of the non-zero eigenvalues decreases with $p$. Crucially, the previous holds  not only for the eigenvalues of the QFIM that were non-zero in the noiseless setting,  but also for those that the noise ``turns on''. This result indicates that the state's sensitivity to parameter changes decreases with increasing noise levels. As we will prove below (see Theorems~\ref{lem:exp-supression} and~\ref{lem:exp-supp-pauli}), this is a general consequence of the presence of noise in a QNN.

\subsection{Global depolarizing noise}
\label{sec:global}

Here we study the overparametrization of general QNNs acting on an $n$-qubit state under a simple noise model: Global depolarizing noise (see Eq.~\eqref{eq:global_dep}).
We henceforth assume that global depolarizing noise channels act before and after every unitary channel in the QNN with the same probability $p$. That is, we consider the case when 
\small
\begin{align}
   \widetilde{\rho}_{\thv} &= \NC^{Depol}\circ\CC^M_{\theta_M}\circ\NC^{Depol}\circ\cdots\circ\NC^{Depol}\circ\CC^1_{\theta_1}\circ\NC^{Depol}(\rho)\nonumber\,.
\end{align}
\normalsize
We remark here that since the noise channel $\NC^{Depol}$ acts before the first parametrized gate in the circuit, $\CC^1_{\theta_1}$, we avoid the change of rank in the quantum state (from a pure state to a full-rank state) that would occur otherwise.
It is not hard to see that the action of the global depolarizing noise channels can be commuted through to the end of the circuit, so that 
\begin{align}
&\NC^{Depol}\circ\CC^M_{\theta_M}\circ\NC^{Depol}\circ\cdots\circ\NC^{Depol}\circ\CC^1_{\theta_1}\circ\NC^{Depol}\nonumber\\
&=\underbrace{\NC^{Depol}\circ\cdots\circ\NC^{Depol}}_{\times (M+1)}\circ \CC^M_{\theta_M}\circ\cdots\circ\CC^1_{\theta_1}\nonumber\\
&=\NC^{Depol}_{eff}\circ \CC^M_{\theta_M}\circ\cdots\circ\CC^1_{\theta_1}\label{eq:eff-depol}\,.
\end{align}
Here, we have defined 
\begin{equation}\label{eq:depol-eff-def}
    \NC_{eff}^{Depol}(\rho)=(1-p)^{M+1}\rho+\left(1-(1-p)^{M+1}\right)\frac{\id}{d}\,.
\end{equation}
This shows that we can express the output state of the QNN as 
\begin{align}
\widetilde{\rho}_{\thv}
&=(1-p)^{M+1}\CC_{\thv}(\rho)+ \left(1-(1-p)^{M+1}\right)\frac{\id}{d}\nonumber\\
&=(1-p)^{M+1}\rho_{\thv}+\left(1-(1-p)^{M+1}\right)\frac{\id}{d}\,.\label{eq:depol-noise-state}
\end{align}

More generally,  we note that the previous result can be extended to the case when global depolarizing noise channels with different layer-dependent probabilities $p_m$ act in between the gates of the circuit. In particular, one now finds
\begin{equation}
    \NC_{eff}^{Depol}(\rho)=\prod_{m=1}^{M+1}(1-p_m)\rho + \left(1-\prod_{m=1}^{M+1}(1-p_m)\right)\frac{\id}{d}\,.
\end{equation}

Next, we study how the rank of the QFIM and the magnitude of its eigenvalues change  due to the presence of global depolarizing noise. In this context, we find it convenient to  first present a useful theorem.
\begin{theorem}\label{theo:noise-end}
Consider the case when a single noise channel acts at the end of the QNN as
\begin{equation}
\widetilde{\CC}_{\thv}=\NC\circ\CC^M_{\theta_M}\circ\cdots\circ\CC^1_{\theta_1}\,.
\end{equation}
 The rank of the QFIM cannot be increased by the action of the noise. That is
\begin{equation}
    \rank[F(\widetilde{\rho}_{\thv})]\leq \rank[F(\rho_{\thv})]\,,
\end{equation}
where $\rho_{\thv}$ and $\widetilde{\rho}_{\thv}$ respectively denote the output states of the noiseless and noisy QNNs (see Eqs.~\eqref{eq:noiseless_state} and~\eqref{eq:noisy-channel}).
\end{theorem}
We refer the reader to Appendix~\ref{ap:A} for a proof of Theorem~\ref{theo:noise-end}. 
The key implication of this theorem is that if the noise acts exclusively at the end of the circuit, then the rank of the QFIM cannot be increased. Hence, it follows that one can overparametrize the noisy QNN, $\widetilde{\CC}_{\thv}$, with the same number of parameters needed to overparametrize the noiseless one, $\CC_{\thv}$.

Using Theorem~\ref{theo:noise-end} we can readily prove the following result for the case of global depolarizing noise.
\begin{theorem}\label{lem:rank-depol}
When global depolarizing channels act before and after every unitary in the QNN, then the rank of the QFIM cannot be increased by the action of the noise. 
\end{theorem}
The proof of this theorem can be found in Appendix~\ref{ap:B}. Theorem~\ref{lem:rank-depol} shows that the presence of global depolarizing channels cannot increase the number of available directions in state space, as the rank of the QFIM is non-increasing. Again, this means that one can overparametrize the noisy QNN with the same number of parameters needed to overparametrize the noiseless QNN, when global depolarizing noise acts throughout the circuit.

In addition, we can also analyze how the eigenvalues of the QFIM change due to the presence of the global depolarizing channels. Here, the following theorem holds. 
\begin{theorem}\label{lem:exp-supression}
When global depolarizing channels act before and after every unitary in the QNN, the entries of the QFIM (and therefore its eigenvalues) satisfy
\begin{equation}
    [F(\rho_{\thv})]_{ij}\in\OC(e^{-p (M+1)})\,
\end{equation}
i.e., they become exponentially suppressed with the product of the number of gates $M$ and the probability of depolarization $p$.  
\end{theorem}
This theorem is proven in Appendix~\ref{ap:C}. Theorem~\ref{lem:exp-supression} indicates that while global depolarizing noise cannot increase the number of available directions in state space, it does suppress the sensitivity of the state to \textit{any} variations in the parameters. This result can be used to further understand the so-called \textit{noise-induced barren plateau} phenomenon~\cite{wang2020noise,franca2020limitations} whereby noise erases all the features in the QML model's training landscape. For instance, let us note that given a linear loss function (i.e., $f_s(x)=x$ in Eq.~\eqref{eq:loss-function}), one has 
\begin{align}
\LC(\thv)&= \Tr[\widetilde{\CC}_{\thv}(\rho)O]\nonumber\\
    &=(1-p)^{M+1} \Tr[\CC_{\thv}(\rho)O]+(1-(1-p)^{M+1}) \Tr[O]\,.\label{eq:lineal-loss-SI}
\end{align}
Equation~\eqref{eq:lineal-loss-SI} shows that the optimization landscape becomes exponentially flat with $M$ and $p$ (hence a noise-induced barren plateau). As such, our results show that noise-induced insensitivities arise already at the level of state space, thus providing a more fundamental understanding of the noise-induced barren plateau phenomenon.

\subsection{Local depolarizing plus unital Pauli channels}

In this section we show that some of the intuition gathered from the previous sections can be extended to more general Pauli noise models. Namely, we here consider how noise affects the QFIM for a QNN acting on $n$-qubits when a fairly general Pauli noise acts. As we will show, the entries of the QFIM, and concomitantly its eigenvalues, get  exponentially suppressed with the product of the number of noise channels and the noise probability. We note that here we will not attempt to prove that, on average, noise increases the rank of the QFIM. This is due to the fact that there  can exist  special types of noise and parameter values for which the rank is not increased (see Sections~\ref{sec:single-qubit} and~\ref{sec:global}). Because of these  subtleties, we will leave a more detailed rank analysis for future work.

In what follows we will consider a general noise model where noise channels are interleaved with the unitary channels of the QNN as
\begin{equation}\label{eq:noisy-channel-1}
\widetilde{\CC}_{\thv}=\NC_{M+1}\circ\CC^M_{\theta_M}\circ\NC_{M}\circ\cdots\circ\NC_2\circ\CC^1_{\theta_1}\circ\NC_1\,,
\end{equation}
for some (potentially layer dependent) noise channels $\NC_m$, $m=1,\ldots,M+1$. Again, we note that $\NC_1$ acts before $\CC^1_{\theta_1}$, which avoids the state changing from pure to mixed after the first parametrized gate. Moreover,  as shown in Fig.~\ref{fig:general_models}, we will assume that each noise channel is composed of local depolarizing noise channels acting on each qubit plus some general unital Pauli noise. That is, 
\begin{equation} \label{eq:noise_channel-2}
    \NC_m = \NC_{loc}^{Dep}(\rho)\circ \NC_m^P(\rho)\,,
\end{equation}
where $\NC_m^P(\rho)$ is an arbitrary unital Pauli quantum channel and $\NC_{loc}^{Dep}$  is a product of local depolarizing channels as given by Eq.~\eqref{eq:local-depol-noise}. For simplicity, we will assume that all  local depolarizing noise channels have the same probability $p$. In Appendix~\ref{ap:extension} we show how our results can be generalized to the case where they have different (qubit- and layer-dependent) probabilities. Moreover, we note that the order in which $\NC_{loc}^{Dep}(\rho)$ and $\NC_m^P(\rho)$ act in Eq.~\eqref{eq:noise_channel-2} will be irrelevant for our purposes, as our results can also be shown to hold when the order is reversed (see Appendix~\ref{ap:extension}).

\begin{figure}[t!]
    \centering
    \includegraphics[width=1\linewidth]{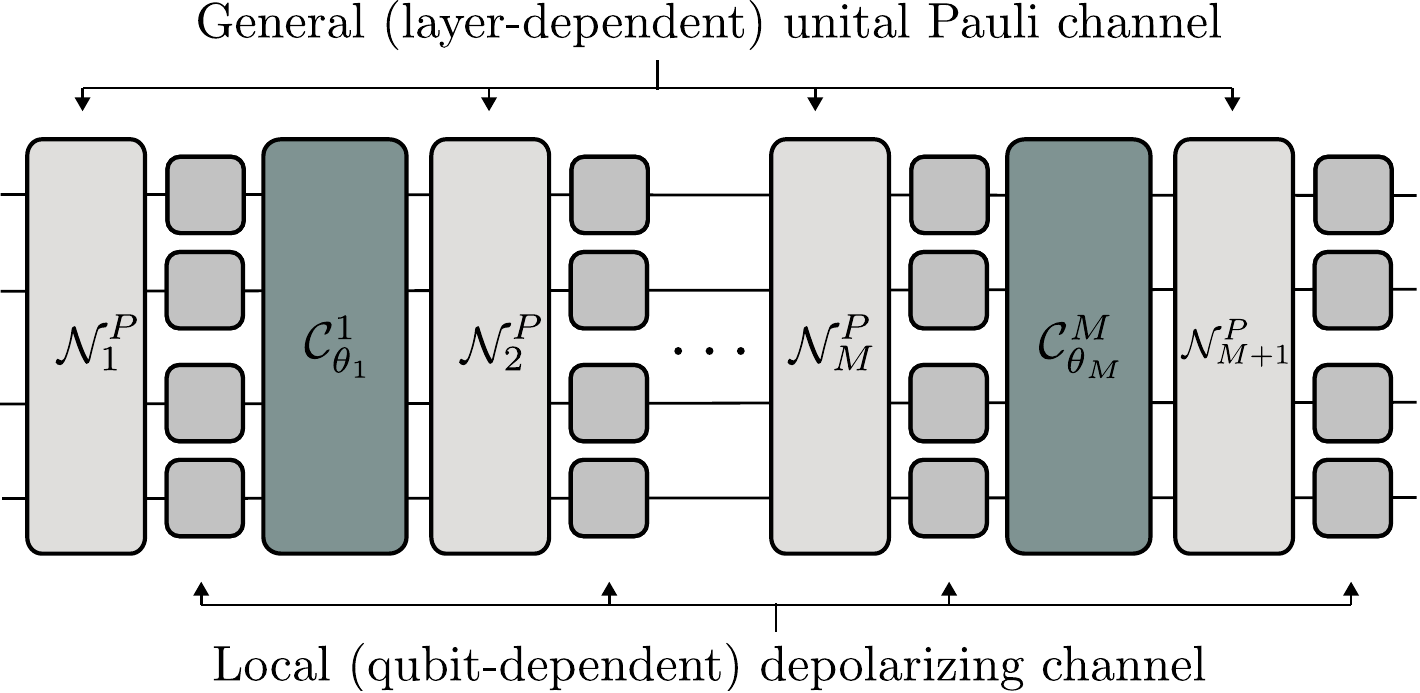}
    \caption{{\bf Schematic representation of a QNN under the general noise model considered.} Our results are derived for a general noise model where unitary gates are interleaved by noise channels composed of local depolarizing noise channels acting on each qubit plus some general unital Pauli noise.  }
    \label{fig:general_models}   
\end{figure}

For this noise model, we prove the following theorem.
\begin{theorem} \label{lem:exp-supp-pauli}
Let $\widetilde{\CC}_{\thv}$ be a noisy channel as in Eqs.~\eqref{eq:noisy-channel-1} and~\eqref{eq:noise_channel-2}, where a Pauli noise channel (composed of a local depolarizing noise acting on each qubit plus a general unital Pauli channel)  acts before and after each gate of the QNN. Furthermore, let $p$ be the  probability of the local depolarizing channels as in Eq.~\eqref{eq:local_depolarizing}. The entries of the QFIM, and thus its eigenvalues, are exponentially suppressed with the product of $M$ and $p$  as $\OC(e^{-2p (M+1)})$.
\end{theorem}
See Appendix~\ref{ap:D} for a proof of Theorem~\ref{lem:exp-supp-pauli}.
This theorem states that under very general noise models, the entries of the QFIM and its eigenvalues vanish exponentially with the noise probability and the number of gates. Crucially, we will have that irrespective of whether the rank of the QFIM is increased or not by the noise, if the circuit is too deep (large $M$), or if the noise levels are too high (large $p$), the state becomes insensitive to parameter changes. Similarly to Theorem~\ref{lem:exp-supression}, this result sheds new light into the noise-induced barren plateau phenomenon~\cite{wang2020noise,franca2020limitations}.

\section{Numerical results}

\begin{figure}[t!]
    \centering
    \includegraphics[width=0.5\textwidth]{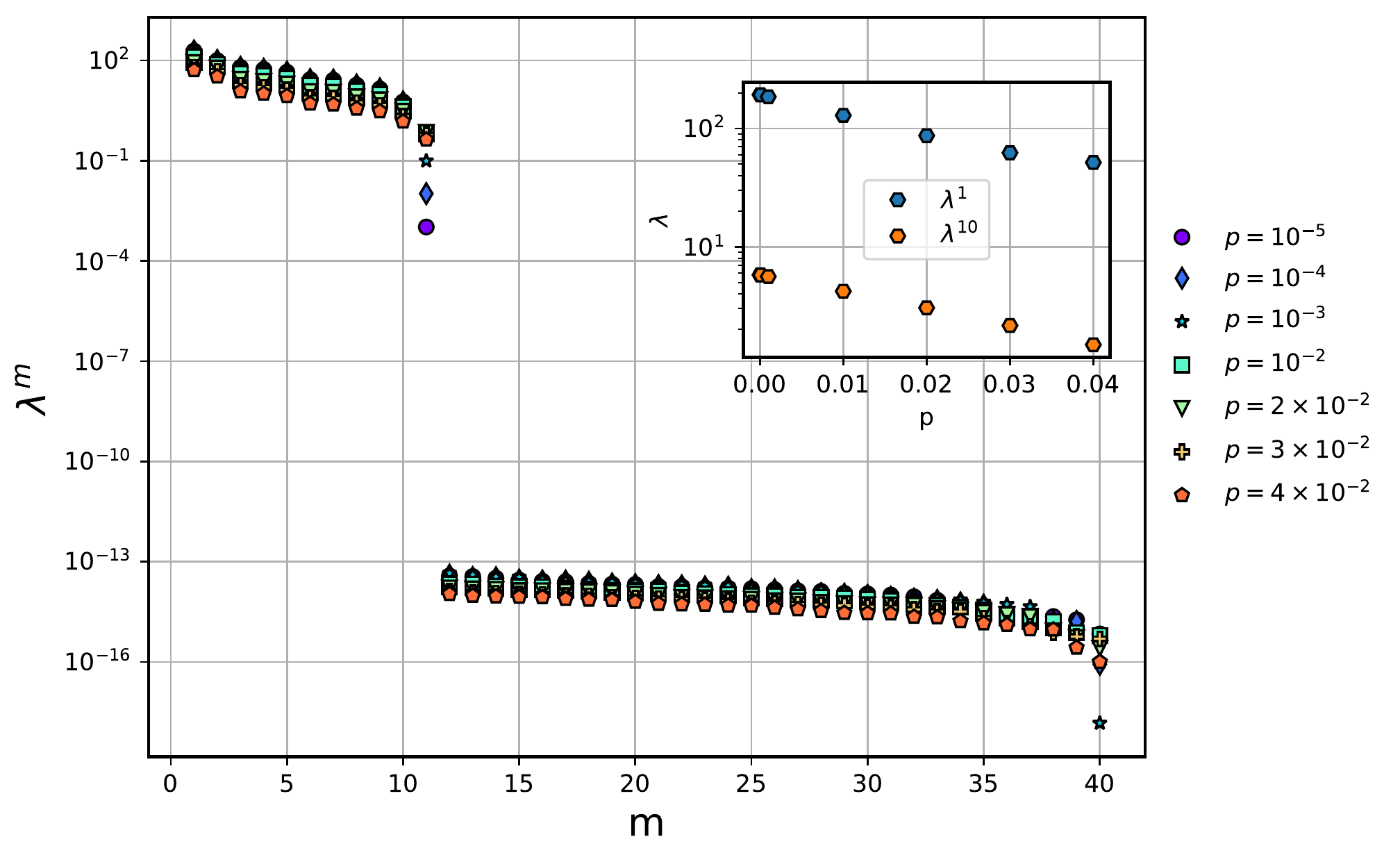}
    \caption{\textbf{Eigenvalues of the QFIM under global depolarizing noise.} Here we consider a problem where an $n=10$ qubit state is sent through an HVA quantum circuit as in Eq.~\eqref{eq:HVA}, with $L=20$ (i.e., with $M=40$ parameters). In the simulations,  global depolarizing noise acts on all qubits before and after each gate. We show the magnitude of the $m$-th eigenvalue of the QFIM for different noise values $p$, at a random point in the landscape. The inset shows the scaling of two non-null eigenvalues with $p$.}
    \label{fig:global1}
\end{figure}

In this section we present numerical results that extend and complement our theoretical findings.  All the simulations presented here have been performed in double precision with the open-source library \texttt{qibo}~\cite{efthymiou2020qibo}, using the fast \texttt{qibojit} backend~\cite{efthymiou2022quantum}. The simulations have been carried out on CPUs, namely IntelCore i7-9750H and AMD Ryzen Threadripper PRO 3955WX cores. 

In particular, we consider the problem where the QNN is given by a  Hamiltonian Variational Ansatz (HVA)~\cite{wecker2015progress,wiersema2020exploring}  with generators inspired by the transverse-field Ising model with periodic boundary conditions. That is, we have $\GC=\{H_0,H_1\}$, with
\begin{equation}
    H_0= \sum_{i=1}^n \sigma_i^z\sigma_{i+1}^z\ ,\quad H_1=\sum_{i=1}^n \sigma_i^x \,,
\end{equation}
and $\sigma_{n+1}^z\equiv \sigma_1^z$. Here, the action of the noiseless QNN  $U(\thv)$ is given by
\begin{equation}\label{eq:HVA}
    U(\thv) = \prod_{l=1}^L e^{-i\theta_{l1} H_1}e^{-i\theta_{l0} H_0}\,,
\end{equation}
where $L$ is the number of layers. Thus, the QNN has $M=2L$ parameters. We have fixed the initial state of the QNN to be the state $\ket{+}^{\otimes n}$. As shown in~\cite{larocca2021theory}, the DLA associated with this ansatz has dimension $\dim(\liea)=\frac{3}{2}n$, meaning that the QNN can be overparametrized with only a polynomial (linear) number of parameters (or layers). In what follows, we will study how the presence of noise affects the QFIM. In all cases, the computations have been carried out at random points in parameter space.

First, in order to validate Theorem~\ref{lem:rank-depol} we have simulated the action of global depolarizing channels acting before and after each gate. The results are depicted in Fig.~\ref{fig:global1}, where the eigenspectrum of the QFIM is plotted for $n=10$ qubits, $M=40$ parameters, and different noise probabilities. Here we see that the rank of the QFIM is unaffected by global depolarizing noise as indicated by Theorem~\ref{lem:rank-depol}. These numerical results also allow us to verify the exponential decrease of the QFIM eigenvalues with the probability of the depolarizing noise, predicted by Theorem~\ref{lem:exp-supression} (see inset in Fig.~\ref{fig:global1}).

\begin{figure}[t!]
    \centering
    \includegraphics[width=0.5\textwidth]{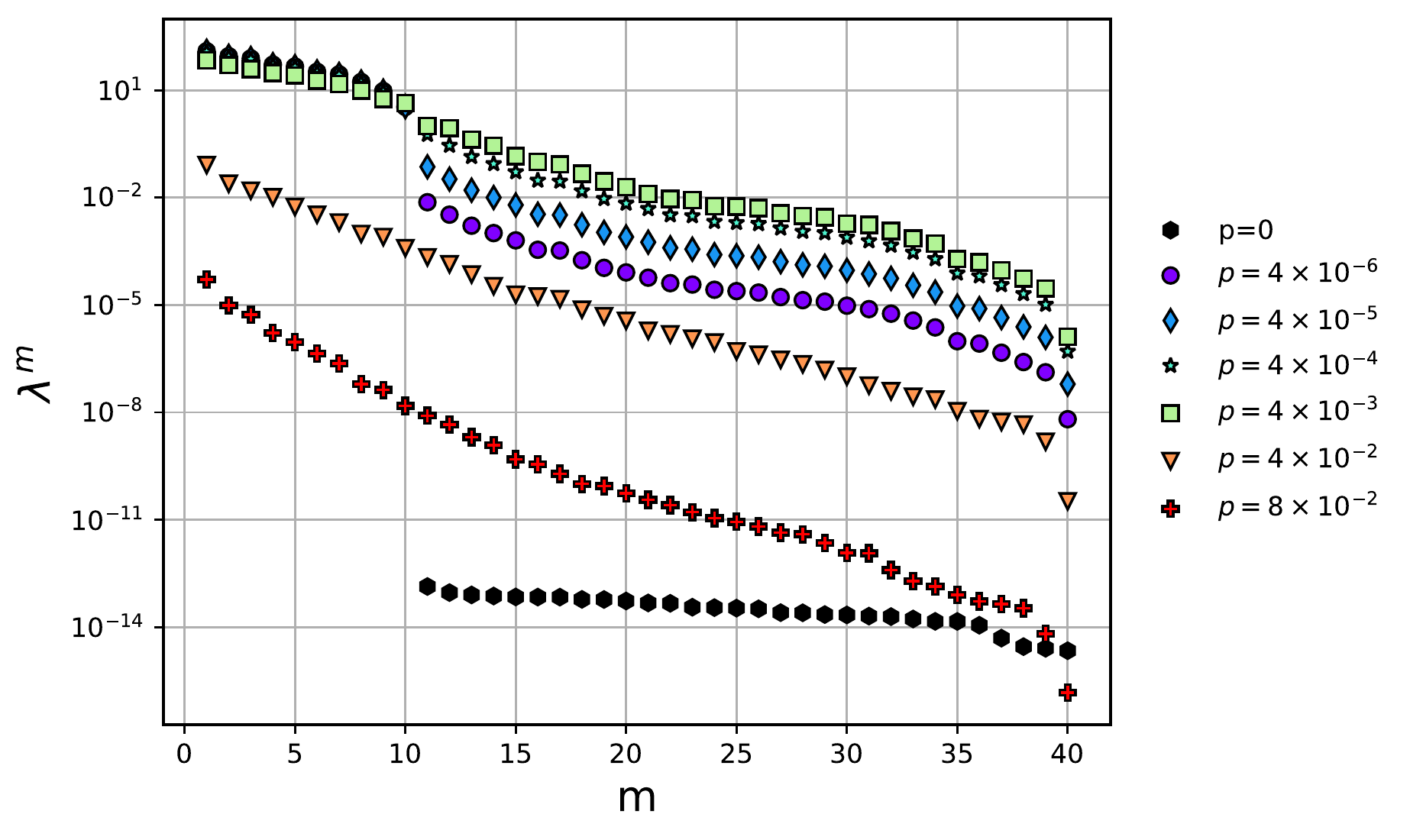}
    \caption{\textbf{Eigenvalues of the QFIM under local depolarizing noise.} Here we consider a problem where an $n=10$ qubit state is sent through an HVA quantum circuit as in Eq.~\eqref{eq:HVA}, with $L=20$ (i.e., with $M=40$ parameters). In the simulations,  local depolarizing noise channels act with the same probability $p$ on all qubits before and after each gate. We show the magnitude of the $m$-th eigenvalue of the QFIM for different noise values $p$, at a random point in the landscape.}
    \label{fig:local1}
\end{figure}

\begin{figure}[t!]
    \centering

    \includegraphics[width=0.5\textwidth]{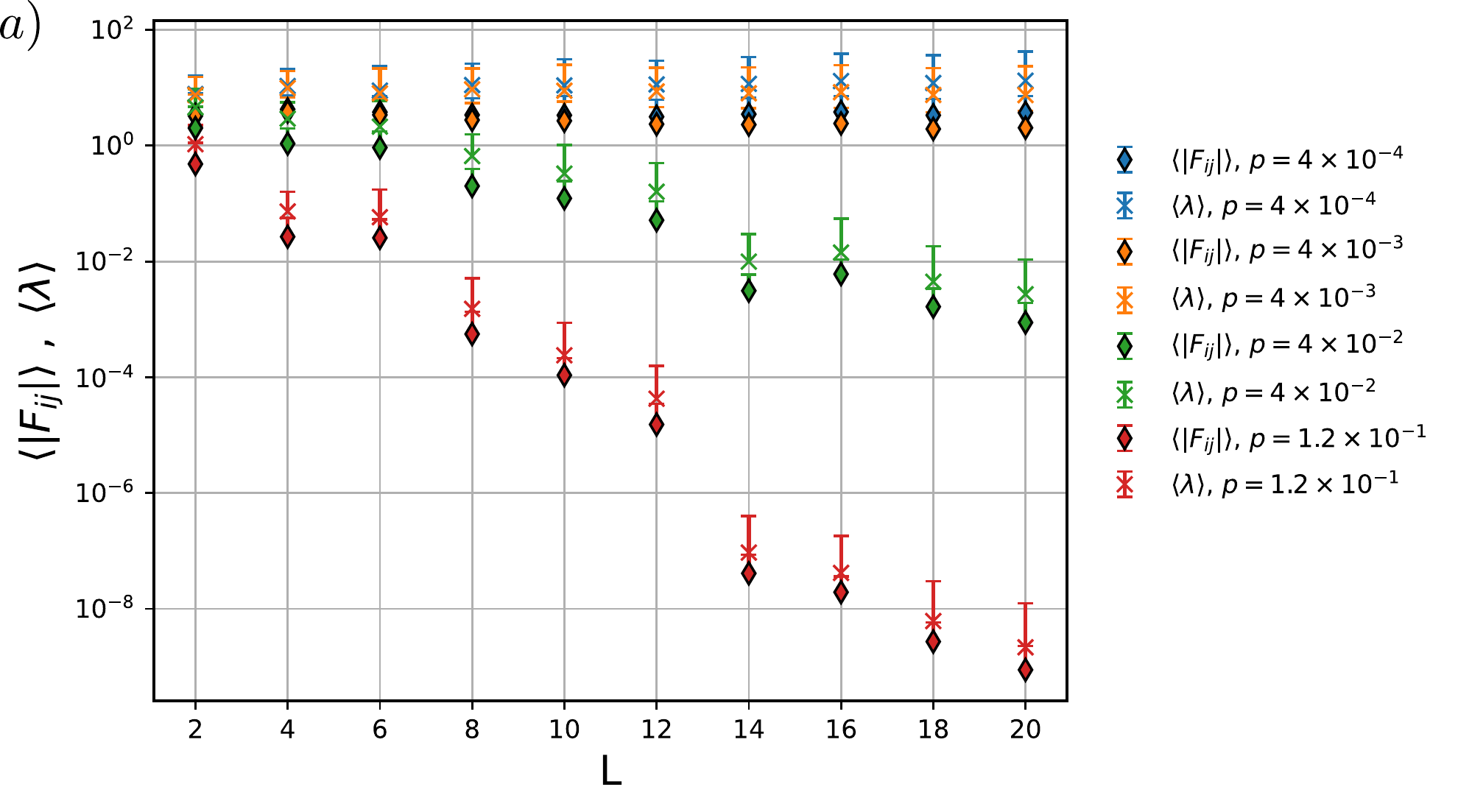}
    
    \includegraphics[width=0.47\textwidth]{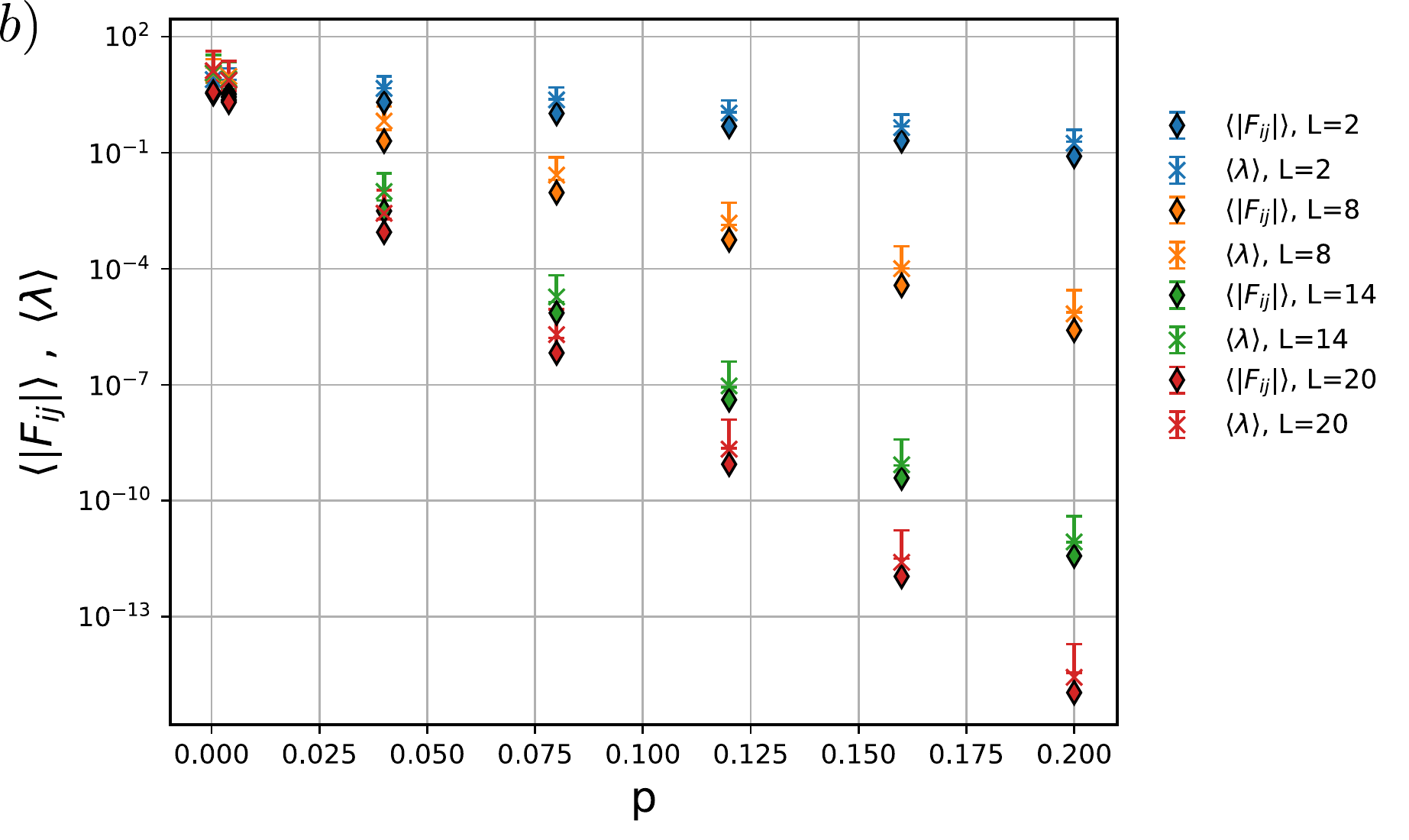}
    \caption{\textbf{Average magnitude of the QFIM entries and eigenvalues in the presence of local depolarizing noise.} Here we consider a problem where an $n=10$ qubit state is sent through an HVA quantum circuit as in Eq.~\eqref{eq:HVA}, with fixed parameter values for each $(L,p)$. In the simulations,  local depolarizing noise channels act with the same probability $p$ on all qubits before and after each gate. We show the average magnitude of the entries of the QFIM and its eigenvalues for different a) number of layers $L$ and b) noise values $p$. Bars depict the standard deviation across the different entries (or eigenvalues) of the QFIM.}
    \label{fig:local2}
\end{figure}

Second, we have simulated the case where local depolarizing channels act on each qubit before and after every gate in the circuit. The results are shown in Fig.~\ref{fig:local1}, where we plot the eigenspectrum of the QFIM for $n=10$ qubits, $M=40$ parameters, and different noise probabilities. In contrast to global depolarizing channels, local noise does increase the rank of the QFIM. This can be observed in the plot from the fact that as soon as $p$ is larger than zero, all the eigenvalues of the QFIM become non-null (as opposed to the noiseless case). As already discussed, this implies that noise enables new directions in state space. Moreover, here we can see that according to Definition~\ref{def:overparametrization}, noise can turn an overparametrized QNN (with saturated rank) into an underparametrized one (where the rank of the QFIM is equal to the number of parameters). Notably, Fig.~\ref{fig:local1} shows that there exists a  certain robustness to noise in the overparametrization phenomenon. This is evidenced from the fact that  when the probability of noise acting is small (e.g., $p\sim10^{-5}$), there still exist a gap of about two orders of magnitude between the dominant eigenvalues and the ``newly-appeared'' ones. Hence, for small noise levels the system can be considered to be in a quasi-overparametrized regime, where large eigenvalues of the QFIM (the ones that were previously non-zero) coexist with small eigenvalues  (the ones that were previously zero).

This  separation in eigenvalue magnitude disappears  when noise levels increase. As shown in Fig.~\ref{fig:local1},  for large enough noise probability (e.g., $p=0.08$) all the eigenvalues  are  exponentially vanishing. Moreover, in  Fig.~\ref{fig:local2}  we show the scaling of the QFIM entries and eigenvalues with the number of gates and noise probability. As in the case of global depolarizing noise, these decrease exponentially (with some statistical fluctuations observed). Taken together, the results in  Figs.~\ref{fig:local1} and~\ref{fig:local2} numerically confirm the result established in Theorem~\ref{lem:exp-supp-pauli}.

\section{Implications to capacity measures}

Let us briefly discuss the implications of the previous results for the capacity measures of Refs.~\cite{haug2021capacity,abbas2020power}. In particular, we have seen that both these measures are related to the maximum rank of the quantum or classical Fisher information matrix. For simplicity, we will consider first a noiseless QNN that has enough parameters $M$ to be well beyond the overparametrization threshold. In this case, the rank of the QFIM, and concomitantly its capacities $D_1(\thv)=\rank[F(\widetilde{\rho}_{\thv})]$ and $D_2$, are saturated, and are such that $D_1(\thv), D_2<M$ (see~\cite{larocca2021theory}).

The results presented in this work indicate that if hardware noise is present, then the rank of the QFIM can increase (e.g., the QFIM can become full rank). Using $D_1(\thv)$ as capacity measure~\cite{haug2021capacity} would imply that the QNN's capacity is increased by noise. Moreover, a similar conclusion can be drawn for the capacity measure of~\cite{abbas2020power} as follows. Indeed, when the noise renders the QFIM full rank, then it becomes invertible and the following inequality holds~\cite{albarelli2019evaluating}
\begin{equation}
    I^{-1}(\widetilde{\rho}_{\thv})\geq F^{-1}(\widetilde{\rho}_{\thv})\,.
\end{equation}
This implies that  the classical Fisher information is also invertible, and thus full rank. Since $\rank[F(\rho_{\thv})]\geq \rank[I(\rho_{\thv})]$~\cite{chan1985hermitian}, this means that when the rank of the QFIM increases and becomes full rank, so does the rank of the classical Fisher information matrix. Hence, using Eq.~\eqref{eq:eff_dim_2} we see that noise can also increase the capacity of the QNN when $D_2$ is used as capacity measure.

We remark that it is  true that new directions are enabled in state space by the action of noise, and that these can be somewhat partially controlled (see Fig.~\ref{fig:directions}). However, it  is also worth recalling that the sensitivity of the noisy state to parameter updates along these direction decreases exponentially with the noise magnitude. In fact, one can expect  that in the regime where the noise is sufficiently large,  the QFIM can be full rank (i.e., $\rank[F(\widetilde{\rho}_{\thv})]=M$) but the magnitude of its eigenvalues   exponentially small (see Theorem~\ref{lem:exp-supp-pauli}). In this scenario, the QNN has a seemingly increased capacity due to noise, but the state is rendered insensitive to parameter changes.

The critical issue here is that the  rank of the QFIM is a discrete number that depends on the number of strictly non-zero eigenvalues of the QFIM, but not on their magnitude. Such issue could be potentially alleviated by considering capacity measures that depend on the magnitude of the eigenvalues.  
For instance, one could  modify the measure of 
~\cite{haug2021capacity} (see  Eq.~\eqref{eq:eff_dim_1}) as 
\begin{equation}D_1^{(\epsilon)}(\thv)=\mathbb{E}\left[\sum_{m=1}^M\ZC^{(\epsilon)}(\lambda^{m}(\thv))\right]\,,
\end{equation}
where $\lambda^{m}(\thv)$ are the eigenvalues of the QFIM for the state $\rho_{\thv}$, $\ZC^{(\epsilon)}(x)=0$ for $x\leq \epsilon$, and $\ZC(x)=1$ for $x>\epsilon$. As such, one would only account for the eigenvalues of the QFIM that are larger than a given tuneable constant $\epsilon$. We leave however the study of such a measure for future work, as more research is needed to understand the interplay between the capacity of QNNs and quantum noise.

\section{Conclusions}

Theoretically understanding the performance of QNNs is a fundamental step to guaranteeing their success in practical realistic scenarios. While there has been tremendous efforts  in studying noiseless QNNs, little is known about their performance when hardware noise acts throughout the computation. However, since noise is a defining property of near-term quantum devices, more research is needed to bridge the gap between our understanding of noiseless and  noisy QNNs.

In this work we focus on the overparametrization of noisy QNNs. 
To analyze this phenomenon, we first present a toy model example which showcases how noise acting throughout a quantum circuit can indeed increase the rank of the QFIM. Crucially, this means that noise can transform an overparametrized QNN into an underparametrized one. Moreover, this toy model also illustrates a second general effect that noise exerts on QNNs, namely it produces a decrease in the magnitude of the QFIM eigenvalues and thus in the sensitivity of the quantum state to parameter updates.

We then derive analytical results proving that a noise channel at the end of a quantum circuit cannot increase the rank of the QFIM. This implies that certain noise models, like global depolarizing noise interleaved with unitary gates, or measurement noise, leave the rank of the QFIM unaffected. In turn, this means that the noisy QNN can be overparametrized with the same number of parameters as the noiseless QNN. However, we also prove that global depolarizing channels suppress the QFIM entries and its eigenvalues exponentially with the product of the number gates and the probability of depolarization. This renders the output of the QNN insensitive to changes in the variational parameters. 

Furthermore, we prove that for fairly general Pauli noise models (consisting of local depolarizing channels and unital Pauli noise), the  eigenvalues and entries of the QFIM get exponentially suppressed with the circuit depth and the noise probability.
 Our results point to a combined effect arising from noise, whereby the rank of the QFIM can be increased, but at the same time the magnitude of all the eigenvalues of the QFIM (both the pre-existing ones and the ones that the noise ``turned on'') get suppressed. Therefore, although noise enables new directions, it also makes the noisy state insensitive to changes in the parameter values.

With the help of numerical simulations we are able to identify three  regimes for the overparametrization phenomenon in the presence of noise. The first corresponds to small noise levels. Here, the magnitude of the new non-zero eigenvalues is very small compared to that of pre-existing ones, whose magnitude remains largely unchanged, indicating a certain robustness to noise. In this ``quasi-overparametrized'' regime, the state is mostly insensitive when the parameters are moved along the directions associated with the newly appeared non-zero eigenvalues. We leave for future work to study the performance and potential for quantum advantage of parametrized quantum circuits in this quasi-overparametrized regime. Indeed, there is evidence that specific types of noise can be useful to improve the trainability of noisy quantum circuits~\cite{sannia2023engineered}. Then, there exists an intermediate regime where the magnitude of the new non-zero eigenvalues is comparable to that of the previous non-zero ones, but smaller than in the noiseless scenario.  That is, in this regime there is no gap between large and small eigenvalues, but rather the eigenvalues (sorted from larger to smaller) lie on a continuous line. Finally, in the third regime, all the eigenvalues vanish and the state becomes (almost) completely insensitive to changes in the parameter values. Moreover,  we find that some of the new directions are purity altering, meaning that the QNN can map the state to regions where  it is more, or less, sensitive to the effects of noise.

We then study the implications of our results to current QNN capacity measures proposed in the literature~\cite{haug2021capacity,abbas2020power}. We find that measures based on the QFIM rank can be misleading when noise is taken into account. In the presence of noise, the QFIM can be transformed from singular to full rank, indicating (according to rank-based measures)
that noise can increase a QNN's capacity. However, the eigenvalues of the QFIM are exponentially suppressed, meaning that the state does not significantly change with parameter updates. This dissonance arises from the fact that the eigenvalue magnitude is not accounted for in rank-based measures, only the number of strictly non-zero eigenvalues. These capacity measures should then be modified accordingly, which is left for future work. 

To conclude, we discuss the impact of our results beyond the overparametrization phenomenon. 
For instance, our results can be understood as shedding new light on the noise-induced barren plateau phenomenon whereby the optimization landscape of QML models gets exponentially flat with the noise (or the depth) of the circuit~\cite{wang2020noise}. Namely, the flatness in the landscape arises from the state being insensitive to changes in the parameters, as evidenced by the exponentially suppressed eigenvalues of the  QFIM. In addition, our results have critical implications to noisy-state quantum sensing~\cite{koczor2020variational,beckey2020variational,kaubruegger2021quantum,huerta2022inference,falaye2017investigating}. Since the ultimate precision achievable for sensing external parameters depends on the quantum Fisher information through the quantum Cram{\'{e}}r-Rao bound~\cite{hayashi2016quantum,liu2016quantum}, our results demonstrate how the utility of a noisy state as a sensor gets degraded by the presence of noise. Finally, an open question that has recently received a lot of attention~\cite{anschuetz2022efficient,goh2023lie,cerezo2023does} is whether QNNs with poly sized DLAs are classically simulable. Although a very relevant question to the present study --poly DLA circuits are the only ones admitting efficient overparametrization-- we want to stress that their efficient classical simulation is not always guaranteed without the access to samples from a quantum computer~\cite{cerezo2023does}. Moreover,  even if shown fully classically simulable, overparametrized QNNs would still find application in problems where one does not seek a computational advantage, e.g., in quantum sensing~\cite{endo2020variational,beckey2020variational}.

\medskip

\acknowledgements 

We acknowledge the Referees of~\cite{larocca2021theory} for pointing us in this fruitful research direction. We would also like to thank  Max Hunter Gordon, Zo\"e Holmes, Eddie Schoute and Fr\'ed\'eric Sauvage for useful conversations. 
D.G-M. was supported by the U.S. Department of Energy, Office of Science, Office of Advanced Scientific Computing Research, under Computational Partnerships program. M.L. acknowledges support by the Center for Nonlinear Studies at Los Alamos National Laboratory (LANL) andy the U.S. Department of Energy (DOE), Office of Science, Office of Advanced Scientific Computing Research, under the Accelerated Research in Quantum Computing (ARQC) program. 
M.C. acknowledges the Quantum Science Center (QSC), a National Quantum Information Science Research Center of the U.S. DOE. This work was also supported by Laboratory Directed Research and Development (LDRD) program of LANL under project numbers 20230049DR and 20230527ECR.
\bibliography{quantum}

\appendix

\section{Proof of Theorem~\ref{theo:noise-end}}
\label{ap:A}

\begin{proof}
    
Recalling that the QFIM is non-increasing under quantum channels, we have that  $F(\widetilde{\rho}_{\thv})\leq F(\rho_{\thv})$, meaning that $F(\rho_{\thv})-F(\widetilde{\rho}_{\thv}) $ is positive semi-definite. Assuming that $F(\rho_{\thv})$ has rank $R=\rank[F(\rho_{\thv})]$, then there exists $M-R$ orthogonal vectors $\{\vec{\omega}\}_{k=1}^{R-M}$ such that 
\begin{equation}
    \vec{\omega}_k^T\cdot F(\rho_{\thv})\cdot\vec{\omega}_k=0\,, \quad \forall k\,.
\end{equation}
These $\{\vec{\omega}_k\}$ vectors form a basis of the null space of $F(\rho_{\thv})$.
From the previous we have that
\begin{equation}\label{eq:pos-sem-proof}
   \vec{\omega}_k^T\cdot ( F(\rho_{\thv})-F(\widetilde{\rho}_{\thv}))\cdot\vec{\omega}_k=-\vec{\omega}_k^T\cdot F(\widetilde{\rho}_{\thv})\cdot\vec{\omega}_k\,,
\end{equation}
But since both $F(\widetilde{\rho}_{\thv})$ and $(F(\rho_{\thv})-F(\widetilde{\rho}_{\thv})) $ are positive semidefinite we must have  $\vec{x}^T\cdot F(\widetilde{\rho}_{\thv})\cdot\vec{x}\geq 0 $  and $\vec{x}^T\cdot ( F(\rho_{\thv})-F(\widetilde{\rho}_{\thv}))\cdot\vec{x}\geq 0 $ for any $\vec{x}$. Combining this realization with Eq.~\eqref{eq:pos-sem-proof} implies that
\begin{equation}
    \vec{\omega}_k^T\cdot F(\widetilde{\rho}_{\thv})\cdot\vec{\omega}_k=0\,.
\end{equation}
In other words, the vectors $F(\widetilde{\rho}_{\thv})$ of the null space of $F(\rho_{\thv})$ are also in the null space of $F(\widetilde{\rho}_{\thv})$. Hence, the action of the noise channel at the end of the QNN cannot increase the rank of the QFIM. 
\end{proof}

\section{Proof of Theorem~\ref{lem:rank-depol}}
\label{ap:B}

\begin{proof}

We recall that the result in Eq.~\eqref{eq:eff-depol} and~\eqref{eq:depol-eff-def} show that adding global depolarizing noise channels at every layer  is equivalent to first acting with the noiseless unitary channel $\CC_{\thv}$ and then with an effective depolarizing noise channel with probability of depolarization of $1-(1-p)^{M+1}$. That is,
\begin{align}\label{eq:equivalency-channels}
&\NC^{Depol}\circ\CC^M_{\theta_M}\circ\NC^{Depol}\circ\cdots\circ\NC^{Depol}\circ\CC^1_{\theta_1}\circ\NC^{Depol}\nonumber\\
&=\NC_{eff}^{Depol}\circ\CC^M_{\theta_M}\circ\cdots\circ\CC^1_{\theta_1}\,,
\end{align}
where we have defined 
\begin{equation}
    \NC_{eff}^{Depol}(\rho)=(1-p)^{M+1}\rho + \left(1-(1-p)^{M+1}\right)\frac{\id}{d}\,.
\end{equation}

Using the right-hand-side of Eq.~\eqref{eq:equivalency-channels} we know that  the action of the depolarizing noise can be pulled to the end of the circuit. Then, the proof follows from Theorem~\ref{theo:noise-end}.
\end{proof}

\section{Proof of Theorem~\ref{lem:exp-supression}}
\label{ap:C}

\begin{proof}
To show that the eigenvalues are exponentially suppressed, we can recall that since the  QFIM is convex, we have
   \begin{equation}\begin{split}
       F(\widetilde{\rho}_{\thv})&=F( (1-p)^{M+1}\rho_{\thv}+(1-(1-p)^{M+1})\frac{\id}{d})\\ &\leq (1-p)^{M+1} F( \rho_{\thv})\,,\end{split}
   \end{equation}
where we have used the fact that $F(\id)=0_{M\times M}$, where $0_M$ denotes the $M\times M$ null matrix. Here we can see that as $M$ increases, $(1-p)^{M+1} F( \rho_{\thv})$ approaches a matrix whose entries are exponentially suppressed by $M$. In the limit $M\rightarrow \infty$, the matrix $(1-p)^{M+1} F( \rho_{\thv})\rightarrow 0_M$.

Rewriting the previous equation we find that $(1-p)^{M+1} F( \rho_{\thv})-F(\widetilde{\rho}_{\thv})\geq 0$, implying that $((1-p)^{M+1} F( \rho_{\thv})-F(\widetilde{\rho}_{\thv}))$ is a positive semi-definite matrix. Respectively denoting as $\{\widetilde{\lambda}^m(\thv)\}$, and $ \{\widetilde{\vec{\lambda}}^m(\thv)\}$ the sets of eigenvalues and eigenvectors of $F(\widetilde{\rho}_{\thv})$, we have
\begin{align}\label{eq:eigen-noisy}
&(\widetilde{\vec{\lambda}}^m(\thv))^T\cdot((1-p)^{M+1} F( \rho_{\thv})-F(\widetilde{\rho}_{\thv}))\cdot\widetilde{\vec{\lambda}}^m(\thv)\nonumber\\
&=(1-p)^{M+1}(\widetilde{\vec{\lambda}}^i(\thv))^T\cdot F( \rho_{\thv})\cdot\widetilde{\vec{\lambda}}^m(\thv)- \widetilde{\lambda}^m(\thv)\geq 0\,.
\end{align}
Then, we can further use the fact that  $\vec{x}^T\cdot  F(\rho_{\thv})\cdot\vec{x}\leq \lambda_{\max}(\thv) $ for any $\vec{x}$, where $\lambda_{\max}(\thv)$ is the largest eigenvalue of the QFIM for the noiseless state $F(\rho_{\thv})$. Combining this result with Eq.~\eqref{eq:eigen-noisy} leads to
\begin{equation}
    \widetilde{\lambda}^m(\thv)\leq (1-p)^{M+1} \lambda_{\max}(\thv)\,,
\end{equation}
or alternatively, $\widetilde{\lambda}^m(\thv)\in\OC((1-p)^{M+1})$. Using that $1-p\leq e^{-p}$, we have $\widetilde{\lambda}^m(\thv)\in\OC(e^{-p(M+1)})$. Hence, we find that the eigenvalues of the noisy-state QFIM become exponentially suppressed with circuit depth and the noise probability.

Next, we can also directly show that the entries of the QFIM are also exponentially suppressed by recalling from Eq.~\eqref{eq:depol-noise-state} that the state at the output of the noisy QNN is 
   $\widetilde{\rho}_{\thv}= (1-p)^{M+1}\rho_{\thv}+(1-(1-p)^{M+1})\frac{\id}{d}$.
This allows us to find that the noisy eigenvalues and eigenvectors change due to the presence of noise as
\begin{equation}
   \widetilde{r}_\mu=(1-p)^{M+1}r_\mu+\frac{(1-(1-p)^{M+1})}{d}\,,  
\end{equation}
and
\begin{equation}
\ket{\widetilde{r}_\mu}=\ket{r_\mu}\, .
\end{equation}
From the previous, we can see that the noisy state QFIM is 
\small
\begin{equation}\begin{split}
    &[F(\widetilde{\rho}_{\thv})]_{ij}\\&=\frac{\sum_{\substack{\mu,\nu\\r_\mu+ r_\nu\neq0}}2(1-p)^{2(M+1)}\Re[\bra{r_\mu}\partial_i\rho_{\thv}\ket{r_\nu}\bra{r_\nu}\partial_j\rho_{\thv}\ket{r_\mu}]}{(1-p)^{M+1}(r_\mu+r_\nu)+\frac{2(1-(1-p)^{M+1})}{d}}\,.\end{split}
\end{equation} 
\normalsize

Note that the scaling with the noise probability $p$ is $\OC\left((1-p)^{M+1}\right)$. Then, we can use the fact that $1-p\leq e^{-p}$, to find $[F(\widetilde{\rho}_{\thv})]_{ij}\in \OC( e^{-p (M+1)})$, which shows that the entries of the QFIM decrease exponentially both with circuit depth and with the noise probability.

\section{Proof of Theorem ~\ref{lem:exp-supp-pauli}}
\label{ap:D}
\begin{proof}

We start by recalling that the QFIM is obtained as a second order expansion of the Bures distance $\BC(\widetilde{\rho}_{\thv},\widetilde{\rho}_{\thv}+\delta)$ in  Eq.~\eqref{eq:bures}. That is,
\begin{equation}\label{eq:second-order}
    \BC(\widetilde{\rho}_{\thv},\widetilde{\rho}_{\thv+\vec{\delta}})=\vec{\delta}^T\cdot F(\rho_{\thv})\cdot \vec{\delta} +\OC(\norm{\vec{\delta}}^3)\,.
\end{equation}
Then, we recall that the trace distance 
\begin{equation}
    \DC(\rho,\sigma)=\frac{1}{2}\norm{(\rho,\sigma)}_1\,,
\end{equation}
is related to the Bures distance by the inequality 
\begin{equation}
    \BC(\rho,\sigma)\leq 2\DC(\rho,\sigma)\,.
\end{equation}
Using the triangle inequality of the trace distance, we have
\begin{equation}
    \BC(\rho,\sigma)\leq 2\left(\DC\left(\rho,\frac{\id}{d}\right)+\DC\left(\sigma,\frac{\id}{d}\right)\right) \,.
\end{equation}
Then, leveraging Pinsker's inequality~\cite{ohya2004quantum} 
\begin{equation}
    \DC(\rho,\sigma)\leq 2\ln(2)S(\rho|\sigma)\,,
\end{equation}
with $ S(\rho|\sigma)$ being the  relative entropy
\begin{equation}
    S(\rho|\sigma)=\Tr[\rho(\log(\rho)-\log(\sigma))]\,,
\end{equation}
leads to
\begin{equation}\label{bures-to-rel-0}
    \BC(\rho,\sigma)\leq 4\ln(2)\left(S\left(\rho\Big|\frac{\id}{d}\right)+S\left(\sigma\Big|\frac{\id}{d}\right)\right)\,.
\end{equation}
Which can be used to find
\begin{equation}\label{bures-to-rel}\small
\BC(\widetilde{\rho}_{\thv},\widetilde{\rho}_{\thv+\vec{\delta}}))\leq 4\ln(2)\left(S\left(\widetilde{\rho}_{\thv}\Big|\frac{\id}{d}\right)+S\left(\widetilde{\rho}_{\thv+\vec{\delta}}\Big|\frac{\id}{d}\right)\right)\,.
\end{equation}

First, let us note that the following lemma holds (see Appendix~\ref{ap:lemm} for a proof).
\begin{lemma}\label{lem:iterative-lemma}
    Let $\NC_{loc}^{Dep}\circ \NC^P$ be a noise channel composed of a general unital Pauli noise channel $\NC^P$ followed by a noise channel $\NC_{loc}^{Dep}$ where  local depolarizing noise channels with probability $p$ act on each qubit as in Eq.~\eqref{eq:local-depol-noise}. Then, we have that 
    \begin{align}
S\left(\NC_{loc}^{Dep}\circ \NC^P(\rho)\Big|\frac{\id}{d}\right)\leq (1-p)^2 S\left(\rho\Big|\frac{\id}{d}\right)\,.
\end{align}
\end{lemma}

Using  Lemma~\ref{lem:iterative-lemma} iteratively $(M+1)$ times we can find 
\begin{equation}
\begin{split}\label{eq:bounds-rel-ent}
    S\left(\widetilde{\rho}_{\thv}\Big|\frac{\id}{d}\right)&\leq (1-p)^{2(M+1)} S\left(\rho\Big|\frac{\id}{d}\right)\\
S\left(\widetilde{\rho}_{\thv+\vec{\delta}}\Big|\frac{\id}{d}\right)&\leq (1-p)^{2(M+1)} S\left(\rho\Big|\frac{\id}{d}\right)\,.
\end{split}
\end{equation}

Combining the result in  Eq.~\eqref{bures-to-rel} with those in Eq.~\eqref{eq:bounds-rel-ent} leads to
\begin{equation}
    \BC(\rho_{\thv},\sigma)\leq 8\ln(2)(1-p)^{2(M+1)}S\left(\rho\Big|    \frac{\id}{d}\right)\,.
\end{equation}
Then, from Eq.~\eqref{eq:second-order} we find that to second order
\begin{equation}
\vec{\delta}^T\cdot F(\widetilde{\rho}_{\thv})\cdot \vec{\delta} \leq 8\ln(2)(1-p)^{2(M+1)}S\left(\rho\Big|    \frac{\id}{d}\right)\,,
\end{equation}
which shows that $\vec{\delta}^T\cdot F(\rho_{\thv})\cdot \vec{\delta}$ is exponentially suppressed as $\OC(p^{M+1})$, because $S\left(\rho\Big|    \frac{\id}{d}\right)\leq \log d$. Since the previous holds for every possible value of $\vec{\delta}$, it necessarily follows that the QFIM matrix is suppressed as $\OC((1-p)^{2(M+1)})$, or concomitantly, as $\OC(e^{-2p(M+1)})$. 

\end{proof}

\section{Proof of Lemma~\ref{lem:iterative-lemma}}\label{ap:lemm}

\begin{proof}
    First, we will leverage the following lemma
    \begin{lemma}[Müller-Hermes/França/Wolf~\cite{muller2016relative}, Theorem 6.1, Rephrased]\label{lem-renyi}
Let $\NC$ be a noisy channel as in Eq.~\eqref{eq:local-depol-noise} where local depolarizing noise channels with probability $p$ acts on each qubit. Then
\begin{equation}
S\left(\NC\left(\rho\right)\Big|\frac{\id}{d}\right)\leq (1-p)^2 S\left(\rho\Big|    \frac{\id}{d}\right)\,.
\end{equation}
\end{lemma}

From Lemma~\ref{lem-renyi} we then have 
\begin{align}
S\left(\NC_{loc}^{Dep}\circ \NC^P(\rho)\Big|\frac{\id}{d}\right)\leq (1-p)^2 S\left(\NC^P(\rho)\Big|\frac{\id}{d}\right)\,.\label{eq:proof-prop1}
\end{align}
Then, recalling that $\NC^P$ is a unital noise channel that maps the identity to the identity, we will have $\NC^P(\id)=\id$, and hence 
\begin{equation}
    S\left(\NC^P(\rho)\Big|\frac{\id}{d}\right)=S\left(\NC^P(\rho)\Big|\frac{\NC^P(\id)}{d}\right)\,.
\end{equation}
Finally, we can use the monotonicity of the relative entropy to find 
\begin{equation}
    S\left(\NC^P(\rho)\Big|\frac{\NC^P(\id)}{d}\right)\leq S\left(\rho\Big|\frac{\id}{d}\right)\,.\label{eq:proof-prop2}
\end{equation}
Combining  Eqs.~\eqref{eq:proof-prop1} and~\eqref{eq:proof-prop1} leads to 
\begin{equation}
S\left(\NC\left(\rho\right)\Big|\frac{\id}{d}\right)\leq (1-p)^2 S\left(\rho\Big|    \frac{\id}{d}\right)\,.
\end{equation}
\end{proof}

\section{Extension of our results for general noise models.}\label{ap:extension}

As mentioned in the main text, we can extend the results in Theorem~\ref{lem:exp-supp-pauli} to the case when the noise model in Eq.~\eqref{eq:noisy-channel-1} is composed of a general unital Pauli channel followed by local depolarizing channels where the depolarization probability is \textit{qubit dependent}. That is, when  
\begin{equation} \label{eq:noise_channel-2-ap}
    \NC_m = \widehat{\NC}_{loc}^{Dep}(\rho)\circ \NC_m^P(\rho)\,,
\end{equation}
and where 
\begin{equation}\label{eq:local-depol-noise-ap}
\widehat{\NC}^{Depol}_{loc}(\rho)=\bigotimes_{j=1}^n \widehat{\NC}^{Depol}_j(\rho)\,,
\end{equation}
with 
\begin{align}\label{eq:pauli_channel_prob-ap}
    \widehat{\NC}^{Depol}_j(\rho)
    &=(1-p_j)\rho + p_j\frac{\id_j\otimes\Tr_{j}[\rho]}{2}\,.
\end{align}
Here, $ 0< p_j < 1$ denotes the qubit-dependent probability of depolarization. As such, the case presented in the main text corresponds to $p_j=p$ for all $j$. As we will se below, the results in Theorem~\ref{lem:exp-supp-pauli} can be generalized as follows
\begin{theorem} \label{lem:exp-supp-pauli-extension}
Let $\widetilde{\CC}_{\thv}$ be a noisy channel as in Eqs.~\eqref{eq:noisy-channel-1} and~\eqref{eq:noise_channel-2-ap}, where a Pauli noise channel (composed of qubit-dependent local depolarizing noise acting on each qubit plus a general unital Pauli channel)  acts before and after each gate of the QNN. Furthermore, let $q=\min_j \{p_j\}$ be the minimum (non-zero) probability of the local depolarizing channels as in Eq.~\eqref{eq:pauli_channel_prob-ap}. The entries of the QFIM, and thus its eigenvalues, are exponentially suppressed with the product of $M$ and $q$  as $\OC(q^{2(M+1)})$.
\end{theorem}

\begin{proof}
    First, let us show that for every noise channel one has
    \begin{equation}
        \widehat{\NC}_{loc}^{Dep}(\rho)\circ \NC_m^P(\rho)=\NC_{loc}^{Dep}(\rho)\circ \widetilde{\NC}_m^{P_eff}(\rho)\,,
    \end{equation}
    where $\NC_{loc}^{Dep}(\rho)$ is a channel composed of local depolarizing channels on each qubit with probability $q=\min_j\{p_j\}$ and where $\NC_m^{P_eff}(\rho)$ is a unital Pauli channel. We start by noting that given a local depolarizing channel with probability $p_j>q$,  one has 
    \begin{equation}
\widehat{\NC}^{Depol}_j(\rho)=\NC^{Depol}_j(\rho)\circ \widetilde{\NC}^{Depol}_j(\rho)\,,
    \end{equation}
where
\begin{equation}
\begin{split}
        \NC^{Depol}_j(\rho)=(1-q)\rho + q\frac{\id_j\otimes\Tr_{j}[\rho]}{2}\,,\\
    \widetilde{\NC}^{Depol}_j(\rho)=(1-\tau_j)\rho + \tau_j\frac{\id_j\otimes\Tr_{j}[\rho]}{2}\,,
\end{split}
\end{equation}
and
\begin{equation}
    \tau_j=\frac{p_j-q}{1-q}\,.
\end{equation}
That is, the local depolarizing channel acting on the $j$-th qubit with probability $p_j>q$ can always be expressed as a concatenation of two depolarizing channels, one with probability $\tau_j$ and another with probability $q$. 

From the previous, we can write 
\begin{align}
\widehat{\NC}^{Depol}_{loc}(\rho)=\bigotimes_{j=1}^n \NC^{Depol}_j(\rho)\circ \bigotimes_{j=1}^n\widetilde{\NC}^{Depol}_j(\rho)\nonumber\\
=\NC_{loc}^{Dep}(\rho)\circ \widetilde{\NC}_{loc}^{Dep}(\rho) \,,
\end{align}
where $\NC_{loc}^{Dep}(\rho)=\bigotimes_{j=1}^n \NC^{Depol}_j(\rho)$ is a noise channel where local depolarizing noise acts on each qubit with probability $q$, and where  $\widetilde{\NC}_{loc}^{Dep}(\rho)=\bigotimes_{j=1}^n\widetilde{\NC}^{Depol}_j(\rho)$ is a  channel where a local depolarizing channel  with probability $\tau_j$ acts on the $j$-th qubit. 
Replacing the previous result in Eq.~\eqref{eq:noise_channel-2-ap} leads to 
\begin{align} 
    \NC_m &= \widehat{\NC}_{loc}^{Dep}(\rho)\circ \NC_m^P(\rho)\nonumber\\
    &=\NC_{loc}^{Dep}(\rho)\circ \widetilde{\NC}_{loc}^{Dep}(\rho)\circ \NC_m^P(\rho)\nonumber\\
    &=\NC_{loc}^{Dep}(\rho)\circ \widetilde{\NC}_m^{P_eff}(\rho)\,,
\end{align}
where we have defined the unital Pauli channel  $\widetilde{\NC}_m^{P_eff}(\rho)=\widetilde{\NC}_{loc}^{Dep}(\rho)\circ \NC_m^P(\rho)$.

From here one can follow the proof of Theorem~\ref{lem:exp-supp-pauli}\,.
\end{proof}

Finally, we note that, as discussed in the main text, Theorem~\ref{lem:exp-supp-pauli} was derived for the case when
\begin{equation} 
    \NC_m = \NC_{loc}^{Dep}(\rho)\circ \NC_m^P(\rho)\,.
\end{equation}
However, the theorem also holds if
\begin{equation} 
    \NC_m = \NC_m^P(\rho)\circ \NC_{loc}^{Dep}(\rho)\,.
\end{equation}
The previous can be seen by noting that the following lemma holds.
\begin{lemma}\label{lem:iterative-lemma-2}
    Let $ \NC^P\circ \NC_{loc}^{Dep}$ be a noise channel composed of a noise channel $\NC_{loc}^{Dep}$ where  local depolarizing noise channels with probability $p$ act on each qubit followed by a  general unital Pauli noise channel $\NC^P$. Then, we have that 
    \begin{align}
S\left(\NC_{loc}^{Dep}\circ \NC^P(\rho)\Big|\frac{\id}{d}\right)\leq (1-p)^2 S\left(\rho\Big|\frac{\id}{d}\right)\,.
\end{align}
\end{lemma}
Where the proof of Lemma~\ref{lem:iterative-lemma-2} follows that of Lemma~\ref{lem:iterative-lemma}. Hence, we can use Lemma~\ref{lem:iterative-lemma-2} iteratively $M+1$ times and recover the same result as that in Theorem~\ref{lem:exp-supp-pauli}.

\end{proof}
\end{document}